\documentclass{llncs}
\usepackage{amsmath}
\usepackage{stmaryrd}
\usepackage{txfonts} 
\usepackage{xspace}
\usepackage[all]{xypic} 

\usepackage{enumitem}
\usepackage{tikz}
\usetikzlibrary{decorations,arrows,shapes,automata}
\usetikzlibrary{fit,matrix}
\tikzset{%
  highlight/.style={rectangle, rounded corners,
draw,
very thick,inner sep=1pt,color=#1!60}
}
\newcommand{\progStore}{\mathsf{store}}
\newcommand{\progOkChange}{\mathsf{check}}
\newcommand{\progsplit}{\mathsf{split}}
\newcommand{\progmerge}{\mathsf{merge}}
\newcommand{\progFlatten}{\mathsf{flatten}}

\newcommand{\atmset}{\mathtt{\mathbb X}}	
\newcommand{\cp}[2]{{#2}^\mathbf{#1}}
\newcommand{\cpr}[2]{\cp{#1}{#2}_R}
\newcommand{\cpw}[2]{\cp{#1}{#2}_W}
\newcommand{\modl}{\mathsf m}	
\newcommand{\mrg}[3]{ ^{#2}_{#3} \triangleright \, #1 }
\newcommand{\pll}{ {||} }							
\newcommand{\splt}[3]{ #1 \triangleleft \, ^{#2}_{#3} }
\newcommand{\readOf}[1]{\mathbb{R}_{#1}}
\newcommand{\readable}[1]{\mathtt{r}_{#1}}
\newcommand{\readset}{\mathsf{Rd}}
\newcommand{\valuset}{\mathsf{V}}
\newcommand{\writable}[1]{\mathtt{w}_{#1}}
\newcommand{\writeset}{\mathsf{Wr}}
\newcommand{\testendo}{?\!\!?}			
\newcommand{\testpdl}{?}				
\newcommand{\writeOf}[1]{\mathbb{W}_{#1}}
\newcommand{\storeset}{\mathsf{St}}

\newcommand{\Dlpa}{\ensuremath{\mathsf{DL\text{-}PA}}\xspace}
\newcommand{\DlpaPll}{\ensuremath{\mathsf{DL\text{-}PA}^\pll}\xspace}
\newcommand{\Pdl}{\ensuremath{\mathsf{PDL}}\xspace}

\newcommand{\assgntop}[1]{{\mathtt {+} #1}}
\newcommand{\assgnbot}[1]{{\mathtt {-} #1}}
\newcommand{\assgntopR}[1]{{\mathtt r {+} #1}}
\newcommand{\assgnbotR}[1]{{\mathtt r {-} #1}}
\newcommand{\assgntopW}[1]{{\mathtt w {+} #1}}
\newcommand{\assgnbotW}[1]{{\mathtt w {-} #1}}
\newcommand{\assgntopV}[1]{{\mathtt {+} #1}}
\newcommand{\assgnbotV}[1]{{\mathtt {-} #1}}
\newcommand{\assgnpropV}[2]{(#1 \testpdl ; \assgntopV{#2}) \ndet (\lnot #1 \testpdl ; \assgnbotV{#2})}
\newcommand{\card}[1]{|#1|}
\newcommand{\eqdef}{\stackrel{\text{def}}{=}}

\newcommand{\intPgm}[1]{\llbracket #1 \rrbracket}
\newcommand{\lbox}[1]{ \big[ #1 \big] }
\newcommand{\ldia}[1]{ \big\langle #1 \big\rangle}
\newcommand{\leqv}{ \leftrightarrow }
\newcommand{\limp}{ \rightarrow }
\newcommand{\ndet}{\,{\cup}\,}
\renewcommand{\phi}{\varphi}
\newcommand{\propset}{\mathbb P}
\newcommand{\propsetOf}[1]{\propset_{#1}}
\newcommand{\modinter}{\cap}
\newcommand{\seqseq}[1]{ \text{\Large ;}_{#1} ~ }
\newcommand{\set}[1]{\{#1\}}
\newcommand{\suchthat}{~ : ~}
\newcommand{\tuple}[1]{ \langle #1 \rangle}

\title{Resource Separation in Dynamic Logic of Propositional Assignments }
\author{Joseph Boudou$^1$, Andreas Herzig$^2$, Nicolas Troquard$^3$}
\institute{
  IRIT, University of Toulouse, France \and
  IRIT, CNRS, France \and
  Free University of Bozen-Bolzano, Italy 
}

\begin{document}
\maketitle

\begin{abstract}
We extend dynamic logic of propositional assignments by adding an operator of parallel composition that is inspired by separation logics. 
We provide an axiomatisation via reduction axioms, thereby establishing decidability. 
We also prove that the complexity of both the model checking and the satisfiability problem stay in PSPACE.
\end{abstract}


\section{Introduction}

It is notoriously delicate to extend Propositional Dynamic Logic \Pdl with an operator of parallel composition of programs. 
Several attempts were made in the literature:
Abrahamson as well as Mayer and Stockmeyer studied a semantics in terms of interleaving \cite{Abrahamson80,MayerS96}; 
Peleg and Goldblatt modified the interpretation of programs from a relation between possible worlds to a relation between
possible worlds and sets thereof \cite{Peleg87,Goldblatt92}; 
Balbiani and Vakarelov studied the interpretation of parallel composition of programs $\pi_1$ and $\pi_2$ as 
the intersection of the accessibility relations interpreting $\pi_1$ and $\pi_2$ \cite{BalbianiV03}. 
However, it seems fair to say that there is still no consensus which of these extensions is the `right' one. 

Dynamic Logic of Propositional Assignments \Dlpa \cite{BalbianiHerzigTroquard-Lics13,BalbianiHST14} 
is a version of Propositional Dynamic Logic (PDL) whose atomic programs are 
assignments of propositional variables $p$ to true or false, respectively written $\assgntop p$ and $\assgnbot p$. 
We and coauthors have shown that many knowledge representation concepts and formalisms can be captured in \Dlpa, such as 
update and revision operations \cite{Herzig-Kr14}, 
database 
repair 
\cite{
FeuilladeHR19}, 
lightweight dynamic epistemic logics \cite{DBLP:conf/atal/CharrierS15,CooperHMMR16,DBLP:conf/atal/CharrierS17},
planning~\cite{HerzigEtal-Ecai14}, and
judgment aggregation \cite{DBLP:journals/logcom/NovaroGH18}. 
The mathematical properties of \Dlpa are simpler than those of PDL, in particular, 
the Kleene star can be eliminated \cite{BalbianiHerzigTroquard-Lics13} and 
satisfiability and model checking are both PSPACE complete~\cite{BalbianiHST14}. 

In this paper we investigate how dynamic logic can be extended with 
a program operator of parallel composition $\pi_1 \pll \pi_2$ of two programs $\pi_1$ and $\pi_2$ 
that is inspired by separation logic. 
The latter was studied in the literature as an account of concurrency,
among others by Brookes and by O'Hearn 
\cite{OHearn04,Brookes04,BrookesO16}.
Their Concurrent Separation Logic is characterised by two main principles:
\begin{enumerate}
\item
When two programs are executed in parallel then 
the state of the system is partitioned (`separated') between the two programs: 
the perception of the state and its modification is viewed as 
being local to each of the two parallel programs. 
Each of them therefore has a partial view of the global state. 
This entails that parallelism in itself does not modify the state of the system:
the parallel execution of two programs that do nothing does not change the state. 
The formula 
$\phi \limp \lbox{ \top ? \pll \top ? } \phi$ 
should therefore be valid, where ``$ \top ? $'' is the test that $\top$ is true (which always succeeds).
\item
The execution of a parallel program $\pi_1 \pll \pi_2$ should be insensitive to 
the way the components of $\pi_1$ and $\pi_2$ are interleaved. 
Hence ``race conditions'' \cite{BrookesO16} must be avoided: 
the execution should not depend on the order of execution of atomic programs in $\pi_1$ and $\pi_2$ 
(where we consider tests to be atomic, too). 
Here we interpret this requirement in a rather radical way: 
when there is a race condition between two programs then they cannot be executed in parallel. 
For example, the parallel program $\assgntopV p \pll \assgnbotV p $
where $\assgntopV p$ makes $p$ true and $\assgnbotV p $ makes $p$ false 
is inexecutable because there is a conflict: 
the two possible interleavings 
$\assgntopV p ; \assgnbotV p $ and $\assgnbotV p  ; \assgntopV p $ are not equivalent. 
We even consider that $ \assgntopV p \pll \assgntopV p$ 
is inexecutable,
which some may consider a bit over-constrained\footnote{This restriction can be related to the fact that
the formula $e \mathrel\mapsto e' \mathbin\ast e \mathrel\mapsto e'$ is unsatisfiable in Separation Logic~\cite{Reynolds02}.}.
\end{enumerate}


In formal frameworks for the verification of parallel programs such as the one proposed by Brookes 
and O'Hearn~\cite{Brookes04,OHearn04}, allowing race conditions is a necessary feature for the 
framework to be able to prove a property of programs, namely that they are race-free. 
On the contrary, dynamic logics permit to prove properties of formulas, and atomic actions are actually even totally abstracted away in most dynamic logics. 
In such abstract settings,
whether two given atomic actions can be executed concurrently is a semantic detail of each model.
For instance, in dynamic logics with a parallel composition based on separation (like in
\cite{DBLP:journals/entcs/BenevidesFV11,Boudou16,DBLP:journals/logcom/BalbianiB18}),
the separation relation of the model provides the possibility to forbid race conditions.
The race condition issue arises in logics based on \Dlpa because atomic actions are concrete:
basically, each atomic action potentially changes the valuation of exactly one propositional variable.
Hence it is natural to consider access to the propositional variables as the main resource. 
In this perspective, a decision has to be made on whether the separation of these resources is strict or not, 
i.e., whether race conditions are allowed or not.
In the present work, we have chosen a strict separation semantics because it is the simplest solution satisfying
the two principles stated above.

We have not yet said what one should understand by a \Dlpa system state. 
A previous approach of one of us only considered the separation of valuations, i.e., 
of truth values of propositional variables \cite{Herzig-Wollic13}. 
Two separating conjunctions in the style of separation logic were defined on such models. 
This however did not allow one to define an adjoint implication as usually done in the separation logic literature, which was somewhat unsatisfactory. 
Another paper that was coauthored by one of us has richer models: 
valuations are supplemented by information about writability of variables \cite{HerzigEtal-Ijcai19}. 
It is supposed that a variable can only be assigned by a program when it is writable. 
Splitting and merging of such models can be defined in a natural way, thus providing a meaningful interpretation of parallel composition. 
When parallel composition is based on separation, writability information permits to resolve merge conflicts.
Consider for instance the executions of the two programs $\top \testpdl \pll \assgntopV p$ and $\assgnbotV p \pll \assgntopV p$ from a state in which $p$ is false.
We argue that intuitively, the former program should lead to a state in which $p$ is true
whereas the latter program should either not be executable
or non-deterministically lead to two possible states, one where $p$ is true and one where $p$ is false.
However, without writability information, the states before the merge of each of these programs turn out to be identical:
$p$ is false in the left branch but true in the right branch.
Writability information allows the merge operation to distinguish these two situations and to resolve the conflict in the former case.

We here push this program further and consider models having moreover information about readability of variables. 
We suppose that writability implies readability\footnote{%
As suggested by one of the reviewers of a previous version of the present paper \cite{DBLP:conf/tap/BoudouHT19}, this constraint may be relaxed 
and one may suppose that a program can modify a variable without being able to read its value. 
This would simplify the presentation of the logic; however, we believe that our inclusion constraint is natural in most applications.
}
and that a variable can only be tested if it is readable. 
Our tests $\phi \testendo$ therefore differ from the standard tests of PDL and \Dlpa in that their 
executability depends on whether the relevant variables are readable. In particular, while 
$\ldia{ p \testendo } \top \limp p$, remains valid, its converse
$p \limp \ldia{ p \testendo } \top $ becomes invalid in our logic: 
it may be the case that $p$ is true but cannot be read.

Distinguishing the variables that can be tested 
permits some useful checks.
Consider for instance the execution of the program $\left(\assgnbotV p ; q \testendo \right) \pll \left(\assgnbotV q ; p \testendo\right)$
from a state in which both $p$ and $q$ are true.
Without readability, this program can be executed and results in a state in which both $p$ and $q$ are false.
However, no interleaving of this program can be executed.
Adding readability of variables permits to detect this issue:
we enforce that a variable cannot be read by one subprogram of a parallel composition if it can be written by the other subprogram.

The paper is organised as follows.
In Section~\ref{sec:models} we define models and the two ternary relations `split' and `merge' on models. 
In Section~\ref{sec:language} we define the language of our logic and 
in Section~\ref{sec:interpretation} we give the interpretation of formulas and programs. 
In Section~\ref{sec:axiomatisation} we axiomatise the valid formulas by means of reduction axioms and
in Section~\ref{sec:complexity} we establish that the satisfiability problem is PSPACE complete. 
Section~\ref{sec:conclusion} sums up our contributions and discusses related work and the application to parallel planning. 
The annex contains a proof of associativity of parallel composition.\footnote{
The present paper is a more elaborate version of 
\cite{DBLP:conf/tap/BoudouHT19}. 
It contains proofs of the results, more motivation and explanations, and furthermore proves that our operator of parallel composition is associative. 
}

\section{Models and Their Splitting and Merging }\label{sec:models} 

Let $\propset$ be a countable set of propositional variables. 
We use $p, q,\ldots$ for elements of $\propset$. 
A model (alias a system state) is a 
triple $\modl = \tuple{\readset,\writeset,\valuset}$ 
where $\readset$, $\writeset$, and $\valuset$ are subsets of $\propset$ such that $\writeset \subseteq \readset$. 
The intuition is that $\readset$ is the set of readable variables, $\writeset$ is the set of writable variables, and $\valuset$ is a valuation: 
its elements are true, while those of its complement $\propset \setminus \valuset$ are false. 
The constraint that $\writeset \subseteq \readset$ means that writability implies readability. 

The special case when $\writeset = \propset$ is typical when checking the validity or the satisfiability of a formula $\phi$.
Lemma~\ref{theo:irrelevantVariables} in Section~\ref{sec:interpretation} will prove that,
as expected, this case is equivalent to the case where $\writeset$ and $\readset$ are the set of propositional variables occurring in $\phi$.
In fact, the writability and readability sets are mostly useful for checking properties of subprograms of parallel compositions.
On the other end of the spectrum, the case $\writeset = \emptyset$ is of little interest since the only executable programs will be those without assignments.

Two models 
$\modl_1 = \tuple{\readset_1,\writeset_1,\valuset_1}$ and 
$\modl_2 = \tuple{\readset_2,\writeset_2,\valuset_2}$ are \emph{RW-disjoint} 
if and only if the writable variables of one model and the readable variables of the other are disjoint, i.e., 
if and only if $\writeset_1 \cap \readset_2  = \writeset_2 \cap \readset_1  = \emptyset $. 
For example, 
$\modl_1 = \tuple{ \set p , \set p , \emptyset}$ and 
$\modl_2 = \tuple{ \set p , \emptyset , \emptyset}$ fail to be RW-disjoint:
in $\modl_1$, some program $\pi_1$ modifying the value of $p$ may be executable, while 
for programs executed in $\modl_2$, the value of $p$ 
may differ depending on whether it is read 
before or after the modification by $\pi_1$ took place. 

As writability implies readability, 
RW-disjointness of $\modl_1$ and $\modl_2$ implies that
$\writeset_1$ and $\writeset_2$ are disjoint. 

We define ternary relations $\triangleleft$ (`split') and $\triangleright$ (`merge') on models as follows:
\begin{center}\begin{tabular}{lll}
$\splt{\modl}{\modl_1} {\modl_2} $ & iff & $\modl_1$ and $\modl_2$ are RW-disjoint, 
$\readset = \readset_1 \cup \readset_2 $, $\writeset = \writeset_1 \cup \writeset_2$,  \\&&  and $\valuset = \valuset_1 = \valuset_2$;
\\
$\mrg{\modl}{\modl_1} {\modl_2} $ & iff  & $\modl_1$ and $\modl_2$ are RW-disjoint, 
$\readset = \readset_1 \cup \readset_2$, $\writeset = \writeset_1 \cup \writeset_2$,  \\&&
$\valuset_1 \setminus \writeset = \valuset_2 \setminus \writeset $, and 
$\valuset = (\valuset_1 \cap \writeset_1) \cup (\valuset_2 \cap \writeset_2) \cup (\valuset_1 \cap \valuset_2) $. 
\end{tabular}\end{center}
For example, for 
$\modl = \tuple{ \readset,\writeset,\valuset }$ and 
$\modl_2 = \tuple{ \readset_2,\writeset_2,\valuset_2 }$ 
we have $\splt{\modl}{\modl} {\modl_2} $ if and only if
$\writeset_2 = \emptyset$, $\valuset_2 = \valuset$, and
$\readset_2 \subseteq \readset \setminus \writeset$.
In particular, 
$\splt{ \tuple{\emptyset,\emptyset,\emptyset} }{\modl_1}{\modl_2}$ if and only if 
$\modl_1 = \modl_2 = \tuple{\emptyset,\emptyset,\emptyset}$. 
Contrarily to splitting, merging does not keep the valuation constant: 
it only keeps constant the non-modifiable part $\valuset \setminus \writeset$ of the valuation $\valuset $ and 
puts the results of the allowed modifications of $\writeset$ together. 
These modifications cannot conflict because $\modl_1$ and $\modl_2 $ are RW-disjoint. 
Figure~\ref{fig:splitmerge} illustrates each of these two operations by an example. 
The checks that are performed in the merge operation are reminiscent of the self composition technique in the analysis of secure information flows~\cite{DarvasEtal05,Scheben2016}.\footnote{%
We are grateful to Rainer H\"ahnle for pointing this out to us. 
}

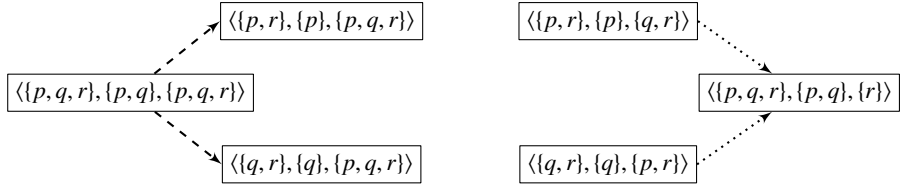
\begin{figure}[t]
  \centering
  \begin{tikzpicture}[>=latex', join=bevel, initial text = , every node/.style=, scale=0.9]
  \node (tosplit) at (0bp, 0bp) [draw] {$\tuple{ \set{p,q,r} , \set{p,q} , \set{p,q,r} } $};

  \node (split1) at (80bp, 30bp) [draw] {$\tuple{ \set{p,r} , \set{p} , \set{p,q,r} } $};
  \node (split2) at (80bp, -30bp) [draw] {$\tuple{ \set{q,r} , \set{q} , \set{p,q,r} } $};

  \node (tomerge1) at (200bp, 30bp) [draw] {$\tuple{ \set{p,r} , \set{p} , \set{q,r} } $};
  \node (tomerge2) at (200bp, -30bp) [draw] {$\tuple{ \set{q,r} , \set{q} , \set{p,r} } $};
  \node (merged) at (280bp, 0bp) [draw] {$\tuple{ \set{p,q,r} , \set{p,q} , \set{r} } $};
  
  \draw[thick, dashed, ->] (tosplit) to node [] {} (split1.west);
  \draw[thick, dashed, ->] (tosplit) to node [] {} (split2.west);

  \draw[thick, dotted, ->] (tomerge1.east) to node [] {} (merged);
  \draw[thick, dotted, ->] (tomerge2.east) to node [] {} (merged);  
  
\end{tikzpicture}
\caption{Examples of split and merge operations:
the left side illustrates the split of the model $\tuple{ \set{p,q,r} , \set{p,q} , \set{p,q,r} } $ into 
$\tuple{ \set{p,r} , \set{p} , \set{p,q,r} } $ and 
$\tuple{ \set{q,r} , \set{q} , \set{p,q,r} } $; 
the right side illustrates the merge of the models
$\tuple{ \set{p,r} , \set{p} , \set{p,q,r} } $ and 
$\tuple{ \set{q,r} , \set{q} , \set{p,q,r} } $ into
 $\tuple{ \set{p,q,r} , \set{p,q} , \set{r} } $.
}
\label{fig:splitmerge} 
\end{figure}

\goodbreak
The set $\readset$ of readable variables of a model $\modl$ 
determines which models cannot be distinguished from $\modl$:
\begin{center}
$\modl \sim \modl'$ \ iff \ $\readset = \readset' , \writeset = \writeset' , \valuset \cap \readset = \valuset' \cap \readset' $.
\end{center}
Hence $\modl $ and $\modl'$ are indistinguishable if 
(1) they have the same readable and writable variables and  
(2) the valuations are identical as far as their readable parts are concerned.
This relation will serve to interpret tests: 
the test $\phi \testendo $ of a formula $\phi$ is conditioned by its truth in all read-indistinguishable models, i.e., 
in all models where the readable variables have the same truth value. 

\section{Language}\label{sec:language}

Formulas and programs are defined by the following grammar,
where $p$ ranges over the set of propositional variables $\propset$:
\begin{align*}
\phi & ::= p \mid \top  \mid  \lnot \phi  \mid  \phi \lor \phi  \mid  \ldia \pi \phi ,
\\
\pi & ::= \assgntopV p \mid \assgnbotV p \mid
		\assgntopR p \mid \assgnbotR p \mid
		\assgntopW p \mid \assgnbotW p \mid
			\phi \testpdl \mid 
			\phi \testendo \mid 
			\pi ; \pi \mid \pi \ndet \pi \mid 
			\pi^\ast \mid \pi \pll \pi .
\end{align*}
The program $\assgntopV p$ makes $p$ true and $\assgnbotV p $ makes $p$ false, where the executability of these two programs is conditioned by the writability of $p$. 
The program $\assgntopR p$ makes $p$ readable and 
$\assgnbotR p $ makes $p$ unreadable; similarly,
$\assgntopW p$ makes $p$ writable and 
$\assgnbotW p$ makes $p$ non-writable.
We suppose that these four programs are always executable. 
The program $\phi \testpdl$ is the PDL test that $\phi$, that we call \emph{exogenous};
$\phi \testendo $ is the \emph{endogenous} test that $\phi$: it is 
conditioned by the readability of the relevant variables of $\phi$. 

The formula $\lbox \pi \phi$ abbreviates $\lnot \ldia \pi \lnot \phi$.
Given an integer $n \geq 0$, the program $\pi^n$ is defined inductively by 
$\pi^0 = \top \testpdl $ and 
$\pi^{n+1} = \pi ; \pi^n $. 
Similarly, $\pi^{\leq n}$ is defined by 
$\pi^{\leq 0} = \top \testpdl $ and 
$\pi^{\leq n+1} = \top \testpdl \ndet (\pi ; \pi^{\leq n}) $. 
For a finite set of variables 
$P = \{p_i\}_{1 \leq  i \leq n}$ 
and associated programs 
$\{ \pi_i(p_i)\}_{1 \leq  i \leq n}$,
we use the notation
$\seqseq{p \in P} \pi(p)$ to denote the sequence
$ \pi_1(p_1) ; \cdots ; \pi_n(p_n)$, in some order. 
We will make use of this notation with care to guarantee that the ordering of the elements of $P$ does not matter. 

The set of propositional variables occurring in a formula $\phi$ 
is noted $\propsetOf \phi $ and 
the set of those occurring in a program $\pi$ is noted $\propsetOf \pi $. 
For example, $\propsetOf{ p \lor \ldia{\assgntopV{q} } \lnot r } = \{p,q,r\}$.

\section{Semantics}\label{sec:interpretation} 


Let $\modl = \tuple{\readset,\writeset,\valuset}$ be a model. 
Formulas are interpreted as sets of models: 
\begin{center}\begin{tabular}{lll}
$\modl \models \top$;
\\
$\modl \models p $ & iff & $p \in \valuset, \text{ for } p \in \propset$;
\\
$\modl \models \lnot \phi $ & iff & $\modl \not \models \phi $;
\\
$\modl \models \phi \lor \psi$ & iff & $\modl \models \phi $ or $\modl \models \psi$;
\\
$\modl \models \ldia \pi \phi $ & iff & there is a model $\modl'$ such that $\modl \intPgm{ \pi } \modl' $ and $\modl' \models \phi$.
\end{tabular}\end{center}
Programs are interpreted as relations on the set of models: 
\begin{align*}
&\modl \intPgm{ \assgntopV{p} } \modl'&\text{ iff }&\readset' = \readset\text{, }\writeset' = \writeset \text{, }\valuset' = \valuset \cup \{p\} \text{, and }p \in \writeset
\\
&\modl \intPgm{ \assgnbotV{p} } \modl'&\text{ iff }&\readset' = \readset\text{, }\writeset' = \writeset \text{, }\valuset' = \valuset \setminus \{p\} \text{, and }p \in \writeset
\displaybreak[0]\\
&\modl \intPgm{ \assgntopR{p} } \modl'&\text{ iff }&\readset' = \readset \cup \{p\} \text{, }\writeset' = \writeset \text{, and }\valuset'= \valuset
\\
&\modl \intPgm{ \assgnbotR{p} } \modl'&\text{ iff }&\readset' = \readset \setminus \{p\} \text{, }\writeset' = \writeset \setminus \{p\}\text{, and }\valuset'= \valuset
\\
&\modl \intPgm{ \assgntopW{p} } \modl'&\text{ iff }&\readset' = \readset \cup \{p\} \text{, }\writeset' = \writeset \cup \{p\} \text{, and }\valuset' = \valuset
\\
&\modl \intPgm{ \assgnbotW{p} } \modl'&\text{ iff }&\readset' = \readset\text{, }\writeset' = \writeset  \setminus \{p\} \text{, and }\valuset' = \valuset
\displaybreak[0]\\
&\modl \intPgm{ \phi \testpdl }\modl'&\text{ iff }&\modl = \modl'\text{ and }\modl \models \phi
\\
&\modl \intPgm{ \phi \testendo }\modl'&\text{ iff }&\modl = \modl'\text{ and }\modl'' \models \phi\text{ for every }\modl''\text{ such that }\modl'' \sim \modl
\displaybreak[0]\\
&\modl \intPgm{ \pi_1 ; \pi_2 } \modl'&\text{ iff }&\text{there is an }\modl''\text{ such that }\modl \intPgm{ \pi_1 } \modl'' \text{ and }
												\modl'' \intPgm{ \pi_2 } \modl'
\\
&\modl \intPgm{ \pi_1 \ndet \pi_2 } \modl'&\text{ iff }&\modl \intPgm{ \pi_1 } \modl'\text{ or }\modl \intPgm{ \pi_2 } \modl'
\\
&\modl \intPgm{ \pi^\ast } \modl'&\text{ iff }&\text{there is an }n \geq 0 \text{ such that }\modl \intPgm{ \pi } ^n \modl'
\\
&\modl \intPgm{ \pi_1 \pll \pi_2 } \modl'&\text{ iff }&\text{there are }\modl_1, \modl_2, \modl'_1, \modl'_2 \text{ such that }
\splt{\modl}{\modl_1} {\modl_2} \text{, }\mrg{\modl'}{\modl'_1} {\modl'_2} \text{,} \\
&&&\modl_1 \intPgm{ \pi_1 } \modl'_1 \text{, }
\readset_1 = \readset'_1 \text{, }\writeset_1 = \writeset'_1 \text{, }\valuset_1 \setminus \writeset_1 = \valuset'_1 \setminus \writeset'_1 \text{, }\\
&&&\modl_2 \intPgm{ \pi_2 } \modl'_2 \text{, }
\readset_2 = \readset'_2 \text{, }\writeset_2 = \writeset'_2 \text{, }\valuset_2 \setminus \writeset_2 = \valuset'_2 \setminus \writeset'_2
\end{align*}
In the interpretation of assignments of atomic formulas we require 
propositional variables to be modifiable, while 
readability and writability can be modified unconditionally.
When a variable is made writable then it is made readable, too, 
in order to guarantee the inclusion constraint on models; 
similarly when a variable is made unreadable. 
The interpretation of parallel composition $\pi_1 \pll \pi_2$ is such that both $\pi_1$ and $\pi_2$ only modify `their' variables. 
More precisely, parallel composition $\pi_1 \pll \pi_2$ of two programs $\pi_1$ and $\pi_2$ 
relates two models $\modl$ and $\modl'$ when the following conditions are satisfied:
(1)~$\modl$ can be split into $\modl_1$ and $\modl_2$; 
(2)~the execution of $\pi_1$ on $\modl_1$ may lead to $\modl_1'$ and 
    the execution of $\pi_2$ on $\modl_2$ may lead to $\modl_2'$;
(3)~$\modl_1'$ and $\modl_2'$ can be merged into $\modl'$. 
Moreover, 
(4) the modifications are legal: $\pi_1$ and $\pi_2$ neither change readability nor writability, and 
each of them only modifies variables that were allocated to it by the split. 

Figure~\ref{fig:ex:parallel} illustrates the interpretation of the parallel program $\assgnbotV p \pll \assgnbotV q$. 
Some more examples follow. 

\begin{figure}[t]
  \centering
  \begin{tikzpicture}[>=latex', join=bevel, initial text = , every node/.style=, scale=0.9]
  \node (tosplit) at (0bp, 0bp) [draw] {$\tuple{ \set{p,q,r} , \set{p,q} , \set{p,q,r} } $};

  \node (split1) at (80bp, 30bp) [draw] {$\tuple{ \set{p,r} , \set{p} , \set{p,q,r} } $};
  \node (split2) at (80bp, -30bp) [draw] {$\tuple{ \set{q,r} , \set{q} , \set{p,q,r} } $};

  \node (tomerge1) at (200bp, 30bp) [draw] {$\tuple{ \set{p,r} , \set{p} , \set{q,r} } $};
  \node (tomerge2) at (200bp, -30bp) [draw] {$\tuple{ \set{q,r} , \set{q} , \set{p,r} } $};
  \node (merged) at (280bp, 0bp) [draw] {$\tuple{ \set{p,q,r} , \set{p,q} , \set{r} } $};
  
  \draw[thick, dashed, ->] (tosplit) to node [] {} (split1.west);
  \draw[thick, dashed, ->] (tosplit) to node [] {} (split2.west);

  \draw[thick, ->] (split1) to node [above] {$\assgnbotV p$} (tomerge1);
  \draw[thick, ->] (split2) to node [above] {$\assgnbotV q$} (tomerge2);
  
  \draw[thick, dotted, ->] (tomerge1.east) to node [] {} (merged);
  \draw[thick, dotted, ->] (tomerge2.east) to node [] {} (merged);

\end{tikzpicture}

\caption{Illustration of an execution of $\assgnbotV p \pll \assgnbotV q$ at the model $\tuple{ \set{p,q,r} , \set{p,q} , \set{p,q,r} } $. 
}
\label{fig:ex:parallel} 
\end{figure}
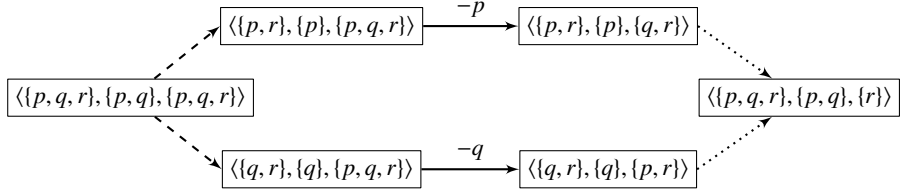

\begin{example}
Suppose  
$\modl = \tuple{\readset,\writeset,\valuset}$ with $\writeset = \readset = \valuset = \{ p, q, r \}$. Then
$\modl' = \tuple{\readset,\writeset,\valuset'}$ with $\valuset' = \{ p, r\}$ is the only model such that 
$\modl \intPgm{  \assgntopV{p} \pll \assgnbotV{q} } \modl' $. 
\end{example}

The next example illustrates the last condition in the interpretation of parallel composition.

\begin{example}
The programs 
$\assgntopV p \pll \assgntopV p $ and
$\assgntopV p \pll (\assgntopW p ; \assgnbotV p ; \assgnbotR p ) $
cannot be executed on the model $\modl = \tuple{ \{p\}, \{p\}, \{p\} }$. 
For the second, suppose there are $\modl_1$ and $\modl_2$ such that $\splt{\modl}{\modl_1} {\modl_2}$ and 
suppose $\assgntopV p$ is executed on $\modl_1$ and 
$\assgntopW p ; \assgnbotV p ; \assgnbotR p$ on $\modl_2$.
For $\assgntopV p$ to be executable we must have $p \in \writeset_1$, and therefore
$p \notin \writeset_2$ (and a fortiori $p \notin \readset_2$) because of the RW-disjointness condition. 
Hence $\modl_1 = \tuple{ \{p\}, \{p\}, \{p\} }$ and $\modl_2 = \tuple{ \emptyset, \emptyset, \{p\} }$.
Then 
$\modl_1' = \modl_1$ is the only model such that $\modl_1 \intPgm{ \assgntopV p } \modl_1' $; and 
$\modl_2' = \tuple{ \emptyset, \emptyset, \emptyset }$ is the only model such that $\modl_2 \intPgm{ \assgntopW p ; \assgnbotV p ; \assgnbotR p } \modl_2' $.
These two models cannot be merged because 
$\valuset_2' \setminus \writeset_2' 
= \emptyset$ 
fails to be equal to 
$\valuset_2 \setminus \writeset_2 = \{p\}$. 
However, $\assgntopV p \pll ( \assgntopW p ; \assgntopV p ; \assgnbotR p )$
is executable on $\modl$: we have
$\splt{\modl}{\modl}{\tuple{\emptyset,\emptyset,\{p\}}}$ and 
$\modl \intPgm{\assgntopV p} \tuple{\{p\},\{p\},\{p\}}$ and
$\tuple{\emptyset,\emptyset,\{p\}} \intPgm{ \assgntopW p ; \assgntopV p ; \assgnbotR p } \tuple{\emptyset,\emptyset,\{p\} } $ and
$\mrg
{ \tuple{\{p\},\{p\},\{p\}} }
{ \tuple{\emptyset,\emptyset,\{p\}} }
{ \tuple{\{p\},\{p\},\{p\}} }$.
\end{example}

It is clear that parallel composition satisfies commutativity. 
It is less obvious that it is also associative. 
The proof is somewhat involved and can be found in the annex. 

Let us finally illustrate the different semantics of the two test operators of our logic.

\begin{example}
Suppose $\modl$ is such that $p \notin \readset$ and $p \in \valuset$. 
Then the program $p \testpdl$ is executable on $\modl$ because $p \in \valuset$.
In contrast, there is no $\modl'$ such that $\modl \intPgm{ p \testendo } \modl'$, 
the reason being that there is always an $\modl''$ such that $\modl \sim \modl''$ and $p \notin \valuset''$,
hence $p \testendo$ is inexecutable. 
\end{example}

In practice, parallel programs should only contain endogenous tests in order to avoid that a subprogram accesses the truth value of a variable that is not among its readable variables.
Actually we have kept PDL tests for technical reasons only:
we could not formulate some of the reduction axioms without them. 

Satisfiability and validity of formulas are defined in the expected way.

\begin{example}
The formulas 
$\phi \limp \lbox{ \top \testendo \pll \top \testendo } \phi$, 
$\lbox{ \assgntopV p \pll \assgnbotV p } \bot$,
$\lbox{ \assgntopV p \pll \assgntopV p } \bot$ and 
$\lbox{ p \testendo \pll \assgntopV p } \bot$ 
whose parallel programs were discussed in the introduction are all valid. 
\end{example}

The formulas $\ldia{ \assgntopV p } \top $ and $\ldia{ \assgnbotV p } \top $ both express that $p$ is writable. 
Moreover, $\ldia{ p \testendo} \top $ expresses that $p$ is true and readable, and 
$\ldia{ \lnot p \testendo} \top $ expresses that $p$ is false and readable;
therefore $\ldia{ p \testendo} \top \lor \ldia{ \lnot p \testendo} \top $ 
expresses that $p$ is readable. 
This will be instrumental in our axiomatisation. 

Finally, for any model $\modl = \tuple{\readset, \writeset, \valuset}$ and $P \subseteq \propset$, we write
$\modl \modinter P$ for the model $\tuple{\readset \cap P, \writeset \cap P, \valuset \cap P}$.
This notation along with the following standard lemma
will be used several times in the remainder of this work.

\begin{lemma}\label{theo:irrelevantVariables}
If there is $P \subseteq \propset$ such that
$\modl_1 \modinter P = \modl_1' \modinter P$,
$\propsetOf \phi \subseteq P$ and $\propsetOf \pi \subseteq P$ then
\begin{enumerate}
  \item $\modl_1 \models \phi$ implies $\modl_1' \models \phi$; and
  \item $\modl_1 \intPgm \pi \modl_2$ implies $\modl_1' \intPgm \pi \modl_2'$,
        for $\modl_2'$ such that $\modl_2' \modinter P = \modl_2 \modinter P$ and
        $\modl_2' \modinter \left( \propset \setminus P\right) =
         \modl_1' \modinter \left( \propset \setminus P\right)$.
\end{enumerate}
\end{lemma}
\begin{proof}[sketch]
The proof is by a straightforward simultaneous
induction on the size of $\phi$ and $\pi$.
\end{proof}

\section{Axiomatisation via Reduction Axioms}\label{sec:axiomatisation}

We axiomatise the validities of our logic by means of reduction axioms, as customary in dynamic epistemic logics \cite{DitmarschHoekKooi07}. 
These axioms transform every formula into a boolean combination of propositional variables
and formulas of the form
$\ldia{ \assgntopV p } \top $ and 
$\ldia{ p \testendo} \top \lor \ldia{ \lnot p \testendo} \top $. 
The former expresses that $p$ is writable: we abbreviate it by $\writable{p}$;
the latter expresses that $p$ is readable: we abbreviate it by $\readable p$. 
Hence we have:
\begin{align*}
\writable{p} &\eqdef \ldia{ \assgntopV p } \top 
\\
\readable p &\eqdef \ldia{ p \testendo} \top \lor \ldia{ \lnot p \testendo} \top 
\end{align*}
The reduction starts by eliminating all the program operators from formulas, where 
the elimination of parallel composition is done by sequentialising it while keeping track of the values of the atoms. 
After that step, the only remaining program operators 
either occur in formulas of the form $\readable p$ or $\writable p$, 
or in modal operators of the form
$\ldia{ \assgntopV p} $,
$\ldia{ \assgnbotV p} $,
$\ldia{ \assgntopR p} $,
$\ldia{ \assgnbotR p} $,
$\ldia{ \assgntopW p} $, or
$\ldia{ \assgnbotW p} $. 
All these modal operators can be distributed over the boolean operators, taking advantage of the fact that all of them are deterministic modal operators 
(validating the Alt$_1$ axiom $\ldia \pi \phi \limp \lbox \pi \phi$). 
Finally, sequences of such modalities facing a propositional variable can be transformed into 
boolean combinations of readability and writability statements $\readable p$ and $\writable p$.
The only logical link between these statements is that 
writability of $p$ implies readability of $p$. 
This is captured by the axiom schema 
$\writable p \limp \readable p$.

The sequentialisation of parallel composition uses copies of variables, so we start by introducing that notion. 
We then define some programs and formulas 
that will allow us to formulate the reduction axioms more concisely. 

\subsection{Copies of Atomic Propositions}\label{sec:copyVars}

Our reduction axioms will introduce fresh copies of each propositional variable, one per occurrence of the parallel composition operator. 
The interpretation of parallel composition being based on separation, each concurrent program operates on its own model.
The copies emulate the separation of models: each concurrent program is executed on a set of copies of propositional variables.

In order to keep things readable we neglect that the copies should be indexed by programs
and denote the copies of the variable $p$ by $\cp k p$, where $\mathbf{k}$ is some integer. 
In principle we should introduce a bijection between the indexes $\mathbf{k}$ and the subprogram they are attached to;
we however do not do so to avoid overly complicated notations.

Given a set of propositional variables $P \subseteq \propset$ and an integer $\mathbf{k} \in \set{1,2}$, we define the set of copies
$\cp{k} P = \{ \cp{k} p \suchthat p \in P\}$. 
Similarly, we define copies of programs and formulas:
the program $\cp{k} \pi$ and the formula $\cp{k} \phi$ are obtained by replacing all their occurrences of
propositional variables $p$ by $\cp k p$. 
For example, $\cp{k} {(\assgntopV p ; q \testendo )}$ equals 
$ \assgntopV{ \cp{k} p } ; \cp{k} q \testendo $.

The following lemma will be instrumental in the soundness proof. 

\begin{lemma}\label{theo:copies}
Let $\pi$ be a program, 
let $\tuple{\readset,\writeset,\valuset}$ be a model, and 
let $\mathbf{k} \in \set{1,2}$. 
Then 
$$ \tuple{\readset,\writeset,\valuset} \intPgm{\pi} \tuple{\readset',\writeset',\valuset'} \text{ iff } 
\tuple{ \cp k {\readset},\cp k {\writeset},\cp k {\valuset}} \intPgm{\cp k \pi} \tuple{\cp k {\readset'},\cp k {\writeset'},\cp k {\valuset'}} . $$
\end{lemma}
\begin{proof}[sketch]
Just as for Lemma~\ref{theo:irrelevantVariables}, the proof is by simultaneous induction on the form of programs and formulas, where the induction hypothesis for the latter is that 
$\tuple{\readset, \writeset, \valuset} \models \phi$ if and only if
$\tuple{\cp k \readset, \cp k \writeset, \cp k \valuset} \models \cp k \phi$.
It relies on the fact that \DlpaPll satisfies the substitution rule.
\end{proof}

Moreover, in order to simulate the semantics of the parallel composition operator our reduction axioms associate to each copy $\cp k p$ two fresh variables $\cpr k p$ and $\cpw k p$,
denoting whether $\cp k p$ was respectively readable and writable just after the split.
Given a set of propositional variables $P$, we define 
$$\storeset(P) =
\set{ \cpr k p \suchthat k \in \set{1,2}, p \in P} \cup
\set{ \cpw k p \suchthat k \in \set{1,2}, p \in P}$$ 
as the set of all these fresh variables.

\subsection{Useful Programs and formulas}\label{sec:usefulFml}

Let $P \subseteq \propset$ be some finite set of propositional variables. 
Table~\ref{fig:useful_programs} lists programs and formulas that will be useful to concisely formulate the reduction axioms.
Observe that the order of the variables in the sequential compositions $ \seqseq{p \in P} ( \cdots ) $ occurring in the above programs does not matter. 
Observe also that the only endogenous tests on the right hand side occur in readability statements $\readable p$. 
(Remember that $\readable p$ abbreviates $\ldia{ p \testendo} \top \lor \ldia{ \lnot p \testendo} \top $.)

\begin{table}[t]
\begin{align*}
\progsplit(P) =&\ \seqseq{p \in P} \Big(
  \assgntopW{\cp 1 p} ; \assgntopW{\cp 2 p} ;
\Big(
  \big( p \testpdl ; \assgntopV{ \cp{1}{p} } ; \assgntopV{ \cp{2}{p} } \big) \ndet 
  \big( \lnot p \testpdl ; \assgnbotV{ \cp{1}{p} } ; \assgnbotV{ \cp{2}{p} } \big) 
\Big) ; \assgnbotR{\cp 1 {p}} ; \assgnbotR{\cp 2 {p}} ;
\\& ~~~~~~~~~~~
\Big(
  \lnot \readable p  \testpdl \ndet
  \big(\writable{p} \testpdl ; ( \assgntopW{\cp 1 {p}} \ndet \assgntopW{\cp 2 {p}} ) \big) 			 \ndet
  \\& ~~~~~~~~~~~~~
  \big(\lnot \writable{p} {\land} \readable p  \testpdl ; \big( \assgntopR{\cp 1 {p}} \ndet \assgntopR{\cp 2 {p}} \ndet 
																			  (\assgntopR{\cp 1 {p}}  ; \assgntopR{\cp 2 {p}}) \big) \big) 
\Big)
\Big)
\\ 
\progStore(P) =&\ \seqseq{p \in P} \seqseq{k \in \{1, 2\}} \Big(
  \assgntopW{\cpr k p} ; \assgntopW{\cpw k p} ;
\\& ~~~~~~~~~~~ 
  \big( \assgnpropV{\readable{\cp k p}}{\cpr k p} \big) ; \big( \assgnpropV{\writable{\cp k p}}{\cpw k p} \big)
\Big)
\\ 
\progOkChange(P) =&\ \bigwedge_{p \in P} \bigwedge_{k \in \{1, 2\}} \Big(
( \cpr k p \leqv \readable{ \cp k {p} } ) 	\land 
( \cpw k p \leqv \writable{ \cp k {p} } ) 	\land 	
( \lnot \cpw k {p} \limp (p \leqv \cp k {p}) )
\Big)
\\ 
\progmerge(P) =&\ \seqseq{p \in P} \Big(\big( 
\lnot \writable{p} \testpdl \ndet 
\\& ~~~~~~~~~~~ 
\big(	(\writable{ \cp{1}{p} } \land \cp{1}{p}) \lor 
		(\writable{ \cp{2}{p} } \land \cp{2}{p}) \testpdl ; \assgntopV p \big) \ndet 
\\& ~~~~~~~~~~~ 
\big(	(\writable{ \cp{1}{p} } \land \lnot \cp{1}{p}) \lor 
		(\writable{ \cp{2}{p} } \land \lnot \cp{2}{p}) \testpdl ; \assgnbotV p \big)
\big) ;
\\& ~~~~~~~~
\seqseq{k \in \{1, 2\}} \big(
  \assgnbotV{\cpr k p} ; \assgnbotV{\cpw k p} ; \assgntopW{\cp k p} ; \assgnbotV{\cp k p} ;
  \assgnbotR{\cpr k p} ; \assgnbotR{\cpw k p} ; \assgnbotR{\cp k p}
\big)
\Big)
\\ 
\progFlatten(\pi_1, \pi_2) =&\ %
  \progsplit\left(\propsetOf{\pi_1 \pll \pi_2}\right) ;
  \progStore\left(\propsetOf{\pi_1 \pll \pi_2}\right) ;
  \cp 1 {\pi_1} ; \cp 2 {\pi_2} ;
  \progOkChange\left(\propsetOf{\pi_1 \pll \pi_2}\right) \testpdl ;
  \progmerge\left(\propsetOf{\pi_1 \pll \pi_2}\right)
\end{align*}
\caption{Useful programs and formulas, for all finite $P \subseteq \propset$ and all programs $\pi_1$ and $\pi_2$.
\label{fig:useful_programs}
}
\end{table}

The $\progsplit(P)$ program simulates the split operation by 
(1) assigning the truth value of every $p \in P$ to its copies $\cp 1 p$ and $\cp 2 p$ and
(2) non-deterministically assigning
two copies of each read and write variable in a way such that a counterpart of the RW-disjointness condition
$\writeset_1 \cap \readset_2 = \writeset_2 \cap \readset_1 = \emptyset$ is guaranteed.  
Note that the assignments $\assgntopW{ \cp{k}{p} }$ also make $\cp{k}{p} $ readable.
The following lemma formally states the main property of $\progsplit(P)$.
It can easily be proved by following the previous observations.

\begin{lemma}\label{lem:progsplit}
For all $P \subseteq \propset$, and all models $\modl$, $\modl_1$ and $\modl_2$ such that
$\modl \modinter P = \modl$,
$$ \splt \modl {\modl_1} {\modl_2} \text{ if and only if }
\modl \intPgm{\progsplit(P)} \modl'$$ with
$\readset' = \readset \cup \cp 1 {\readset_1} \cup \cp 2 {\readset_2}$,
$\writeset' = \writeset \cup \cp 1 {\writeset_1} \cup \cp 2 {\writeset_2}$, and
$\valuset' = \valuset \cup \cp 1 {\valuset_1} \cup \cp 2 {\valuset_2}$.
\end{lemma}

The $\progStore(P)$ program stores the readability and writability states of the
copies of the propositional variable into some fresh variables.
These variables are then used only in $\progOkChange(P)$.

The formula $\progOkChange(P)$ compares the current state with the state just after the split.
It is true if and only if
(1) readability values are identical, 
(2) writability values are identical, and
(3)~truth values are identical for non-writable variables.

The $\progmerge(P)$ program simulates the merge operation by reinstating all those 
read- and write-atoms that had been allocated to the first subprogram in the sequentialisation.
The following lemma formally states the main property of $\progmerge(P)$.
This lemma is weaker than Lemma~\ref{lem:progsplit} for $\progsplit(P)$.
The additional hypotheses are guaranteed to hold by the interplay of programs $\progsplit(P)$ and $\progStore(P)$,
and formula $\progOkChange(P)$.

\begin{lemma}\label{lem:progmerge}
Let $\modl$, $\modl_1$ and $\modl_2$ be models, and $P$ a set of propositional variables
such that $\modl \modinter P = \modl$,
        $\modl_1$ and $\modl_2$ are RW-disjoint,
        $\readset = \readset_1 \cup \readset_2$,
        $\writeset = \writeset_1 \cup \writeset_2$, and
        $\valuset_1 \setminus \writeset = \valuset_2 \setminus \writeset$.
Then
$$\mrg \modl {\modl_1} {\modl_2} \text{ if and only if }
\modl' \intPgm{\progmerge(P)} \modl$$ with
$\readset' = \readset \cup \cp 1 {\readset_1} \cup \cp 2 {\readset_2} \cup \storeset(P)$,
$\writeset' = \writeset \cup \cp 1 {\writeset_1} \cup \cp 2 {\writeset_2} \cup \storeset(P)$, and
$\valuset' = \valuset^\sharp \cup \cp 1 {\valuset_1} \cup \cp 2 {\valuset_2}$
where $\valuset^\sharp$ is any subset of $P \cup \storeset(P)$ such that
$\valuset^\sharp \setminus \writeset = \valuset \setminus \writeset$.
\end{lemma}
\begin{proof}[sketch]
It suffices to prove that
$\modl' \intPgm{\progmerge(P)} \modl$ if and only if
$\valuset = (\valuset_1 \cap \writeset_1) \cup (\valuset_2 \cap \writeset_2) \cup (\valuset_1 \cap \valuset_2) $,
which is straightforward.
\end{proof}

Finally, the $\progFlatten(\pi_1, \pi_2)$ program emulates the execution of the program $\pi_1 \pll \pi_2$ as
a sequential composition of the previous programs.
Notice that there is no occurrence in $\progFlatten(\pi_1, \pi_2)$ of the parallel composition operator,
except possibly inside $\pi_1$ and $\pi_2$.
Lemma~\ref{lem:pllequivalence} below states that the emulation is faithful.
We first need Lemma~\ref{lem:progClean},
which can be seen as an adaptation of Lemma~\ref{theo:irrelevantVariables} to $\progFlatten(P)$.

\begin{lemma}\label{lem:progClean}
For all models $\modl$ and $\modl'$, and all programs $\pi_1$ and $\pi_2$,
$$
\modl \intPgm{\progFlatten(\pi_1,\pi_2)} \modl' \text{ iff }
(\modl \modinter \propsetOf{\pi_1 \pll \pi_2}) \intPgm{\progFlatten(\pi_1, \pi_2)} (\modl' \modinter \propsetOf{\pi_1 \pll \pi_2}).
$$
\end{lemma}
\begin{proof}[sketch]
It suffices to observe that all variables in $\propsetOf{\progFlatten(\pi_1, \pi_2)} \setminus \propsetOf{\pi_1 \pll \pi_2}$ are 
(1) made writable and initialized by $\progsplit$ or $\progStore$, and 
(2) set to false and made unreadable by $\progmerge$.
\end{proof}

We can now prove our main lemma.

\begin{lemma}\label{lem:pllequivalence}
For all models $\modl$ and $\modl'$, and all programs $\pi_1$ and $\pi_2$,
$$
\modl \intPgm{\pi_1 \pll \pi_2} \modl' \text{ if and only if }
\modl \intPgm{\progFlatten(\pi_1, \pi_2)} \modl' \text{.}
$$
\end{lemma}
\begin{proof}[sketch]
Let $P = \propsetOf{\pi_1 \pll \pi_2}$. By Lemmas~\ref{theo:irrelevantVariables} and~\ref{lem:progClean},
we can assume that $\modl \modinter P = \modl$ and $\modl' \modinter P = \modl'$.

For the left-to-right direction, suppose there are models $\modl_1$, $\modl_2$, $\modl'_1$ and $\modl'_2$ such that
$\splt{\modl}{\modl_1} {\modl_2} $, $\mrg{\modl'}{\modl'_1} {\modl'_2} $,
$\modl_1 \intPgm{ \pi_1 } \modl'_1$, 
$\modl_2 \intPgm{ \pi_2 } \modl'_2$, 
$\readset_1 = \readset'_1 $, $\writeset_1 = \writeset'_1 $, $\valuset_1 \setminus \writeset_1 = \valuset'_1 \setminus \writeset'_1 , $
$\readset_2 = \readset'_2 $, $\writeset_2 = \writeset'_2 $, and $\valuset_2 \setminus \writeset_2 = \valuset'_2 \setminus \writeset'_2 $.
The following statements can be proved:
\begin{enumerate}
  \item\label{pllequivalence:ltr:split}
        $\modl \intPgm{\progsplit(P)} \modl^+_1$ with
        $\readset^+_1 = \readset \cup \cp 1 {\readset_1} \cup \cp 2 {\readset_2}$,
        $\writeset^+_1 = \writeset \cup \cp 1 {\writeset_1} \cup \cp 2 {\writeset_2}$, and
        $\valuset^+_1 = \valuset \cup \cp 1 {\valuset_1} \cup \cp 2 {\valuset_2}$.
        The proof relies on Lemma~\ref{lem:progsplit}.
  \item\label{pllequivalence:ltr:store}
        $\modl^+_1 \intPgm{\progStore(P)} \modl^+_2$ with
        $\readset^+_2 = \readset^+_1 \cup \storeset(P)$,
        $\writeset^+_2 = \writeset^+_1 \cup \storeset(P)$,
        $\valuset^+_2 = \valuset^+_1 \cup \valuset_{\storeset}$, and
        $\valuset_{\storeset} =
        \set{ \cpr k p \suchthat \cp k p \in \readset^+_1} \cup
        \set{ \cpw k p \suchthat \cp k p \in \writeset^+_1}$.
        Notice that
        $\valuset_{\storeset} =
        \set{ \cpr k p \suchthat p \in \readset_k} \cup
        \set{ \cpw k p \suchthat p \in \writeset_k}$.
  \item\label{pllequivalence:ltr:pi1}
        $\modl^+_2 \intPgm{\cp 1 {\pi_1}} \modl^+_3$ with
        $\readset^+_3 = \readset \cup \cp 1 {\readset'_1} \cup \cp 2 {\readset_2} \cup \storeset(P)$,
        $\writeset^+_3 = \writeset \cup \cp 1 {\writeset'_1} \cup \cp 2 {\writeset_2} \cup \storeset(P)$, and
        $\valuset^+_3 = \valuset \cup \cp 1 {\valuset'_1} \cup \cp 2 {\valuset_2} \cup \valuset_{\storeset}$.
        The proof relies on Lemmas~\ref{theo:irrelevantVariables} and~\ref{theo:copies}.
  \item\label{pllequivalence:ltr:pi2}
        $\modl^+_3 \intPgm{\cp 2 {\pi_2}} \modl^+_4$ with
        $\readset^+_4 = \readset \cup \cp 1 {\readset'_1} \cup \cp 2 {\readset'_2} \cup \storeset(P)$,
        $\writeset^+_4 = \writeset \cup \cp 1 {\writeset'_1} \cup \cp 2 {\writeset'_2} \cup \storeset(P)$, and
        $\valuset^+_4 = \valuset \cup \cp 1 {\valuset'_1} \cup \cp 2 {\valuset'_2} \cup \valuset_{\storeset}$.
  \item\label{pllequivalence:ltr:check}
        $\modl^+_4 \models \progOkChange(P)$.
  \item\label{pllequivalence:ltr:merge}
        $\modl^+_4 \intPgm{\progmerge(P)} \modl'$.
        The proof relies on Lemma~\ref{lem:progmerge}.
  \item $\modl \intPgm{\progFlatten(\pi_1, \pi_2)} \modl'$.
\end{enumerate}

For the right-to-left direction, let us suppose that there are
models $\modl^+_1$, $\modl^+_2$, $\modl^+_3$ and $\modl^+_4$ such that
$\modl \intPgm{\progsplit(P)} \modl^+_1 \intPgm{\progStore(P)} \modl^+_2 \intPgm{\pi_1} \modl^+_3
\intPgm{\pi_2 ; \progOkChange(P)?} \modl^+_4 \intPgm{\progmerge(P)} \modl'$.
The following statements can be proved:
\begin{enumerate}
  \item\label{pllequivalence:rtl:split}
        $\splt{\modl}{\modl_1}{\modl_2}$ with
        $\readset_1 = \set{ p \suchthat \cp 1 p \in \readset^+_1}$,
        $\writeset_1 = \set{ p \suchthat \cp 1 p \in \writeset^+_1}$,
        $\valuset_1 = \set{ p \suchthat \cp 1 p \in \valuset^+_1}$,
        $\readset_2 = \set{ p \suchthat \cp 2 p \in \readset^+_1}$,
        $\writeset_2 = \set{ p \suchthat \cp 2 p \in \writeset^+_1}$, and
        $\valuset_2 = \set{ p \suchthat \cp 2 p \in \valuset^+_1}$.
        The proof relies on Lemma~\ref{lem:progsplit}.
  \item\label{pllequivalence:rtl:pi1}
        $\modl_1 \intPgm{\pi_1} \modl'_1$ with
        $\readset'_1 = \set{ p \suchthat \cp 1 p \in \readset^+_3}$,
        $\writeset'_1 = \set{ p \suchthat \cp 1 p \in \writeset^+_3}$, and
        $\valuset'_1 = \set{ p \suchthat \cp 1 p \in \valuset^+_3}$.
        The proof relies on Lemmas~\ref{theo:irrelevantVariables} and~\ref{theo:copies}.
  \item\label{pllequivalence:rtl:pi2}
        $\modl_2 \intPgm{\pi_2} \modl'_2$ with
        $\readset'_2 = \set{ p \suchthat \cp 2 p \in \readset^+_4}$,
        $\writeset'_2 = \set{ p \suchthat \cp 2 p \in \writeset^+_4}$, and
        $\valuset'_2 = \set{ p \suchthat \cp 2 p \in \valuset^+_4}$.
  \item\label{pllequivalence:rtl:check}
        $\readset_1 = \readset'_1 $, $\writeset_1 = \writeset'_1 $,
        $\valuset_1 \setminus \writeset_1 = \valuset'_1 \setminus \writeset'_1$,
        $\readset_2 = \readset'_2 $, $\writeset_2 = \writeset'_2 $, and
        $\valuset_2 \setminus \writeset_2 = \valuset'_2 \setminus \writeset'_2$.
        The proof relies on the fact that $\modl^+_4 \models \progOkChange(P)$, and
        $\valuset^+_4 \cap \storeset(P) =
        \set{ \cpr k p \suchthat p \in \readset_k} \cup
        \set{ \cpw k p \suchthat p \in \writeset_k}$.
  \item\label{pllequivalence:rtl:merge}
        $\mrg{\modl'}{\modl'_1}{\modl'_2}$.
        The proof relies on Lemma~\ref{lem:progmerge}.
  \item $\modl \intPgm{\pi_1 \pll \pi_2} \modl'$.
\end{enumerate}
\end{proof}

\subsection{Reduction Axioms for Program Operators}\label{sec:redax_pgmop} 

The reduction axioms for program operators are in Table~\ref{fig:redax_pgmops}.
Those for sequential and non-deterministic composition and for exogenous tests (PDL tests) are as in PDL. 
The one for endogenous tests $ \phi \testendo$ 
checks whether $\phi$ remains true for any possible value of the non-readable variables of $\phi$. 
That for the Kleene star is familiar from \Dlpa. 
That for parallel composition $\pi_1 \pll \pi_2$ executes $\pi_1$ and $\pi_2$ in sequence:
it starts by splitting up readability and writability between the two programs,
then executes $\pi_1$, checks whether $\pi_1$ didn't change the readability and writability variables
and whether all truth value changes it brought about are legal, 
and finally executes $\pi_2$ followed by the same checks for $\pi_2$.

\begin{table}[t]
\begin{align*}
\ldia{\phi \testpdl } \psi \leqv &\ \psi \land \phi
\\
\ldia{\phi \testendo } \psi \leqv &\ \psi \land \lbox{ ~ \seqseq{p \in \propsetOf \phi} \big(
\readable{p} \testpdl \ndet \big( \lnot \readable{p} \testpdl ; (\assgntopV{p} \ndet \assgnbotV{p}) \big) 
\big) } \phi
\\
\ldia{\pi_1 ; \pi_2}  \phi \leqv &\ \ldia{\pi_1 } \ldia{\pi_2}  \phi 
\\
\ldia{\pi_1 \ndet \pi_2}  \phi \leqv &\ \ldia{\pi_1 } \phi \lor \ldia{\pi_2}  \phi 
\\
\ldia{\pi^\ast}  \phi \leqv &\ \ldia{\pi^{\leq 2^{\card{\propsetOf{\phi}}}} }  \phi   
\\
\ldia{\pi_1 \pll \pi_2}  \phi \leqv &\ \ldia{ \progFlatten(\pi_1, \pi_2) } ~ \phi 
\end{align*}
\caption{Reduction axioms for program operators
\label{fig:redax_pgmops}
}
\end{table}

Observe that the validity of the reduction axiom for endogenous tests relies on the fact that the copies  
$\cp 1 p$ and $\cp 2 p$ that are introduced by the program 
$\progsplit( \propsetOf{\pi_1 \pll \pi_2} ) $ are fresh. 
%
%
The length of the right hand side can be shortened by restricting 
$\propsetOf{\pi_1 \pll \pi_2}$ to the propositional variables that are assigned by $\propsetOf{\pi_1 \pll \pi_2}$, i.e., 
to elements $p \in \propset$ such that $\assgntopV p$ or $\assgnbotV p$ occurs in $\propsetOf{\pi_1 \pll \pi_2}$.

The exhaustive application of the equivalences of Table~\ref{fig:redax_pgmops} from the left to the right 
results in formulas whose program operators are either endogenous tests occurring in a readability statement
$\readable p = \ldia{p \testendo} \top \lor \ldia{\lnot p \testendo} \top$, or 
assignments of the form
$\assgntopR p $, $\assgnbotR p $,
$\assgntopW p $, $\assgnbotW p $, 
$\assgntopV p $, or $\assgnbotV p $. 

\subsection{Reduction Axioms for Boolean Operators}\label{sec:redax_atmpgm_bool} 

We now turn to modal operators $\ldia \pi$ where $\pi$ is an atomic assignment, i.e., $\pi$ is of the form 
$\assgntopR p $, $\assgnbotR p $,
$\assgntopW p $, $\assgnbotW p $, 
$\assgntopV p $, or $\assgnbotV p $. 
They are deterministic and can therefore be distributed over the boolean operators.
The corresponding reduction axioms are in Table~\ref{fig:redax_booleanops}.

\begin{table}[t]
\begin{align*}
\ldia{\assgntopV p} \top \leqv &\  \writable p
& \ldia{\assgnbotV p} \top \leqv &\  \writable p
\\
\ldia{\assgntopR p} \top \leqv &\ \top
& \ldia{\assgnbotR p} \top \leqv &\ \top
\\
\ldia{\assgntopW p} \top \leqv &\ \top
& \ldia{\assgnbotW p} \top \leqv &\ \top
\\
\ldia{\assgntopV p } \lnot \phi \leqv &\ 
					\writable{p} \land \lnot \ldia{\assgntopV p } \phi  
& \ldia{\assgnbotV p } \lnot \phi \leqv &\ 
					\writable{p} \land \lnot \ldia{\assgnbotV p } \phi 
\\
\ldia{\assgntopR p } \lnot \phi \leqv &\ 
					\lnot \ldia{\assgntopR p } \phi  
& \ldia{\assgnbotR p } \lnot \phi \leqv &\ 
					\lnot \ldia{\assgnbotR p } \phi 
\\
\ldia{\assgntopW p } \lnot \phi \leqv &\ 
					\lnot \ldia{\assgntopW p } \phi  
& \ldia{\assgnbotW p } \lnot \phi \leqv &\ 
					\lnot \ldia{\assgnbotW p } \phi 
\\
\ldia{\assgntopV p } (\phi \lor \psi) \leqv &\ \ldia{\assgntopV p } \phi \lor \ldia{\assgntopV p } \psi 
& \ldia{\assgnbotV p } (\phi \lor \psi) \leqv &\ \ldia{\assgnbotV p } \phi \lor \ldia{\assgnbotV p } \psi 
\\
\ldia{\assgntopR p } (\phi \lor \psi) \leqv &\ \ldia{\assgntopR p } \phi \lor \ldia{\assgntopR p } \psi 
& \ldia{\assgnbotR p } (\phi \lor \psi) \leqv &\ \ldia{\assgnbotR p } \phi \lor \ldia{\assgnbotR p } \psi 
\\
\ldia{\assgntopW p } (\phi \lor \psi) \leqv &\ \ldia{\assgntopW p } \phi \lor \ldia{\assgntopW p } \psi 
& \ldia{\assgnbotW p } (\phi \lor \psi) \leqv &\ \ldia{\assgnbotW p } \phi \lor \ldia{\assgnbotW p } \psi 
\end{align*}
\caption{Reduction axioms for boolean operators
\label{fig:redax_booleanops}
}
\end{table}

In the first equivalence $\ldia{\assgntopV p} \top \leqv \  \writable p$, 
the right hand side is nothing but an abbreviation of the left hand side. 
We nevertheless state it in order to highlight that the exhaustive application of these reduction axioms 
results in sequences of atomic assignments facing 
either $\writable p$ or $\readable q$. 
These sequences are going to be reduced in the next step.

\subsection{Reduction Axioms for Assignments}\label{sec:redax_atmpgm} 

\begin{table}[t]
\begin{align*}
\ldia{\assgntopV p } q \leqv &\ \begin{cases} 
					\writable{p} 	& \text{ if } q = p \\
					\writable{p} \land q & \text{ otherwise } 
					\end{cases} 
& \ldia{\assgnbotV p } q \leqv &\ \begin{cases} 
					\bot 	& \text{ if } q = p \\
					\writable{p} \land q & \text{ otherwise } 
					\end{cases} 
\\
\ldia{\assgntopR p } q \leqv &\ q
& \ldia{\assgnbotR p } q \leqv &\ q
\\
\ldia{\assgntopW p } q \leqv &\ q
& \ldia{\assgnbotW p } q \leqv &\ q
\\
\ldia{\assgntopV p} \readable q \leqv &\  \writable p \land \readable q
& \ldia{\assgnbotV p} \readable q \leqv &\  \writable p \land \readable q
\\
\ldia{\assgntopR p} \readable q \leqv &\ \begin{cases}
								\top & \text{ if } q = p \\
								\readable q & \text{ otherwise }
								\end{cases}
& \ldia{\assgnbotR p} \readable q \leqv &\ \begin{cases}
								\bot & \text{ if } q = p \\
								\readable q & \text{ otherwise }
								\end{cases}
\\
\ldia{\assgntopW p} \readable q \leqv &\ \begin{cases}
                \top & \text{ if } q = p \\
                \readable q & \text{ otherwise }
                \end{cases}
& \ldia{\assgnbotW p} \readable q \leqv &\ \readable q
\\
\ldia{\assgntopV p } \writable{q} \leqv &\ \writable{p} \land \writable{q}
& \ldia{\assgnbotV p } \writable{q} \leqv &\ \writable{p} \land \writable{q}
\\
\ldia{\assgntopR p } \writable{q} \leqv &\ \writable{q}
& \ldia{\assgnbotR p } \writable{q} \leqv &\ \begin{cases}
                \bot & \text{ if } q = p \\
                \writable{q} & \text{ otherwise } 
								\end{cases}
\\
\ldia{\assgntopW p } \writable{q} \leqv &\ \begin{cases}
								\top & \text{ if } q = p \\
								\writable{q} & \text{ otherwise }
								\end{cases}
& \ldia{\assgnbotW p } \writable{q} \leqv &\ \begin{cases}
								\bot & \text{ if } q = p \\
								\writable{q} & \text{ otherwise }
								\end{cases}
\end{align*}
\caption{Reduction axioms for assignments
\label{fig:redax_assignments}
}
\end{table}

When atomic programs face propositional variables or readability and writability statements 
then the modal operator can be eliminated (sometimes introducing a writability statement $\writable{p}$). 
The reduction axioms doing that are in Table~\ref{fig:redax_assignments}.

As announced, the exhaustive application of the above axioms
results in boolean combinations of propositional variables and readability and writability statements. 

\subsection{Soundness, Completeness, and Decidability} 

Let us call \DlpaPll our extension of \Dlpa with parallel composition. 
Its axiomatisation is made up of 
\begin{itemize}
\item
an axiomatisation of propositional logic;
\item
the equivalences of Sections~\ref{sec:redax_pgmop}, 
\ref{sec:redax_atmpgm_bool}, 
and~\ref{sec:redax_atmpgm};
\item
the inclusion axiom schema $\writable{p} \limp \readable p$, which is an abbreviation of the formula
$\ldia{ \assgntopV p } \top \limp \big( \ldia{ p \testendo} \top \lor \ldia{ \lnot p \testendo} \top \big) $;
\item
the rule of equivalence for the modal operator 
``from $\phi \leqv \psi$ infer $\ldia \pi \phi \leqv \ldia \pi \psi $''.
\end{itemize}

\begin{theorem}\label{theo:axiomatisationsound}
The axiomatisation of \DlpaPll is sound:
if $\phi$ is provable with the axiomatics of \DlpaPll then it is \DlpaPll valid.
\end{theorem}
\begin{proof}[sketch]
We have to show that the inference rules preserve validity and the axioms are valid. 
For the reduction axiom for endogenous test, it suffices to observe that
$\modl \intPgm{\assgntopV p \ndet \assgnbotV p} \modl'$ if and only if $\modl \sim \modl'$.
The proof of validity of the reduction axiom for Kleene star can easily be adapted from the one in~\cite{BalbianiHerzigTroquard-Lics13}.
The case of the reduction axiom for parallel composition is handled by Lemma~\ref{lem:pllequivalence}.
All other cases are straightforward.
\end{proof}

\begin{theorem}
The axiomatisation of \DlpaPll is complete:
if $\phi$ is \DlpaPll valid then it is provable in the axiomatics of \DlpaPll. 
\end{theorem}
\begin{proof}
The reduction axioms of Sections~\ref{sec:redax_pgmop},~\ref{sec:redax_atmpgm_bool}, and~\ref{sec:redax_atmpgm} 
allow us to transform any formula into an equivalent boolean combination of 
propositional variables and readability and writability statements. 
(Their application requires the rule of replacement of equivalents, which is derivable 
because we have rules of equivalence for all the connectives of the language, 
in particular the above $RE(\ldia \pi)$.)
Let $\phi$ be the resulting formula. 
Then $\phi$ has a \DlpaPll model if and only if 
$$\phi \land 
\bigwedge_{p \in \propset} \big( \writable{p} \limp \readable p \big) $$
has a model in propositional logic, where in propositional logic,
$\readable p$ and $\writable{p}$ are considered to be arbitrary propositional variables; so
there is a priori no connection between them nor with the propositional variable $p$. 
\end{proof}

Based on the reduction of \DlpaPll formulas to boolean formulas 
(and the transformation of $\readable p$ and $\writable{p}$ from abbreviations into propositional variables),
we may check the satisfiability of \DlpaPll formulas by means of propositional logic SAT solvers. 
This is however suboptimal because the reduction may result in a formula that is super-exponentially longer than the original formula. 
In the next section we explore another route.

\section{Complexity via Translation into \Dlpa}\label{sec:complexity}

We establish PSPACE complexity of \DlpaPll satisfiability and model checking by translating formulas and programs to 
Dynamic Logic of Propositional Assignments \Dlpa. 
The language of the latter is the fragment of that of \DlpaPll:
it has neither endogenous tests, nor readability and writability assignments, nor parallel composition. 
Hence the language of \Dlpa is built by the following grammar: 
\begin{align*}
\phi & ::= p \mid \top \mid \lnot \phi \mid (\phi \lor \phi) \mid \ldia \pi \phi
\\
\pi & ::= \assgntopV p \mid \assgnbotV p \mid
			\phi \testpdl \mid 
			(\pi ; \pi) \mid (\pi \ndet \pi) \mid 
			\pi^\ast 
\end{align*} 
None of the operators of the language refers to the $\readset$-component or the $\writeset$-component of models. 
The interpretation of \Dlpa formulas and programs therefore only requires a valuation $\valuset$. 

Our translation from \DlpaPll to \Dlpa eliminates endogenous tests and parallel composition. 
This is done in a way that is similar to their reduction axioms of Table~\ref{fig:redax_pgmops}.
It moreover transforms readability and writability statements into special propositional variables $\readable p$ and $\writable p$, similar to the reduction axioms of Table~\ref{fig:redax_assignments}.

To make this formal, let the set of \emph{atomic formulas} be
$$ \atmset = \propset \cup \{ \writable{p} \suchthat p \in \propset \} \cup \{ \readable p \suchthat p \in \propset \} . $$
Given a set of propositional variables $P \subseteq \propset$, 
$\readOf P = \{ \readable p \suchthat p \in P \}$ 
is the associated set of read-variables and 
$\writeOf P = \{ \writable{p} \suchthat p \in P \}$ 
is the associated set of write-variables. 
Hence $\atmset = \propset \cup \readOf \propset \cup \writeOf \propset$. 
As before, the set of propositional variables occurring in a formula $\phi$ is noted $\propsetOf \phi $ and 
the set of those occurring in a program $\pi$ is noted $\propsetOf \pi $. 
This now includes the $p$'s in $\readable p$ and $\writable{p}$. 
For example, $\propsetOf{ p \land \ldia{\assgntopV{w_q} } \lnot \readable p } = \{p,q\}$. 

We translate the \DlpaPll programs
$\assgntopR{p}$, 
$\assgnbotR{p}$, 
$\assgntopW{p}$, and 
$\assgnbotW{p}$ 
into the \Dlpa programs 
$\assgntop{ \readable{p}}$,
$\assgnbot{ \readable{p}}$,
$\assgntop{ \writable{p}}$ and
$\assgnbot{ \writable{p}}$.
Moreover, we have to `spell out' that 
$\assgnbot{ \writable p }$ has side effect $\assgnbot{ \readable p }$ and that  
$\assgntop{ \readable p }$ has side effect $\assgntop{ \writable p }$. 
Hence the programs and formulas of Table~\ref{fig:useful_programs} become the \Dlpa programs and formulas listed in Table~\ref{fig:useful_dlpa}.
\begin{table}[t]
\begin{align*}
\progsplit(P) =&\ \seqseq{p \in P} \Big( 
\Big(
  \big( p \testpdl ; \assgntopV{ \cp{1}{p} } ; \assgntopV{ \cp{2}{p} } \big) \ndet 
  \big( \lnot p \testpdl ; \assgnbotV{ \cp{1}{p} } ; \assgnbotV{ \cp{2}{p} } \big) 
\Big) ;
\\& ~~~~~~~~~~~
\assgnbot{ \writable{\cp 1 {p}}} ; \assgnbot{ \readable{\cp 1 {p}}} ; \assgnbot{ \writable{\cp 2 {p}}} ; \assgnbot{ \readable{\cp 2 {p}}} ;
\\& ~~~~~~~~~~~
\Big(
  \lnot \readable p  \testpdl \ndet 
  \\& ~~~~~~~~~~~~~
  \big(\writable{p} \testpdl ; ( (\assgntop{ \readable{\cp 1 {p}}} ; \assgntop{ \writable{\cp 1 {p}}}) \ndet (\assgntop{ \readable{\cp 2 {p}}} ; \assgntop{ \writable{\cp 2 {p}}}) ) \big) 			 \ndet
  \\& ~~~~~~~~~~~~~
  \big(\lnot \writable{p} {\land} \readable p  \testpdl ; \big( \assgntop{ \readable{\cp 1 {p}}} \ndet \assgntop{ \readable{\cp 2 {p}}} \ndet 
(\assgntop{ \readable{\cp 1 {p}}}  ; \assgntop{ \readable{\cp 2 {p}}}) \big) \big) 
\Big)
\Big)
\displaybreak[0]
\\ 
\progStore(P) =&\ \seqseq{p \in P} \seqseq{k \in \{1, 2\}} \Big(
  \big( \assgnpropV{\readable{\cp k p}}{\cpr k p} \big) ;
  \\& ~~~~~~~~~~~~~
  \big( \assgnpropV{\writable{\cp k p}}{\cpw k p} \big)
\Big)
\displaybreak[0]
\\ 
\progOkChange(P) =&\ \bigwedge_{p \in P} \bigwedge_{k \in \{1, 2\}} \Big(
( \cpr k p \leqv \readable{ \cp k {p} } ) 	\land 
( \cpw k p \leqv \writable{ \cp k {p} } ) 	\land 	
( \lnot \cpw k {p} \limp (p \leqv \cp k {p}) )
\Big)
\displaybreak[0]
\\ 
\progmerge(P) =&\ \seqseq{p \in P} \Big(\big( 
\lnot \writable{p} \testpdl \ndet 
\\& ~~~~~~~~~~~ 
\big(	(\writable{ \cp{1}{p} } \land \cp{1}{p}) \lor 
		(\writable{ \cp{2}{p} } \land \cp{2}{p}) \testpdl ; \assgntopV p \big) \ndet 
\\& ~~~~~~~~~~~ 
\big(	(\writable{ \cp{1}{p} } \land \lnot \cp{1}{p}) \lor 
		(\writable{ \cp{2}{p} } \land \lnot \cp{2}{p}) \testpdl ; \assgnbotV p \big)
\big) ;
\\& ~~~~~~~~
\seqseq{k \in \{1, 2\}} \big(
  \assgnbotV{\cpr k p} ; \assgnbotV{\cpw k p} ; \assgnbotV{\cp k p} ;
  \assgnbotV{\writable{\cp k p}} ; \assgnbotV{\readable{\cp k p}}
\big)
\Big)
\displaybreak[0]
\\ 
\progFlatten(\pi_1, \pi_2) =&\ %
  \progsplit\left(\propsetOf{\pi_1 \pll \pi_2}\right) ;
  \progStore\left(\propsetOf{\pi_1 \pll \pi_2}\right) ;
  \cp 1 {\pi_1} ; \cp 2 {\pi_2} ;
  \progOkChange\left(\propsetOf{\pi_1 \pll \pi_2}\right) \testpdl ;
  \progmerge\left(\propsetOf{\pi_1 \pll \pi_2}\right)
\end{align*}
\caption{Adaptation of the programs and formulas of Table~\ref{fig:useful_programs} to the translation into \Dlpa.
\label{fig:useful_dlpa}
}
\end{table}
Notice that
readability and writability statements are no longer \DlpaPll abbreviations, but are now \Dlpa propositional variables. 

Given a \DlpaPll program or formula, its translation into \Dlpa basically follows the reduction axiom for endogenous tests $\testendo$ and parallel composition $ \pll $ of Table~\ref{fig:redax_pgmops}. 
We replace:
\begin{enumerate}
\item
all occurrences of $\phi \testendo $ with 
$\left[\seqseq{p \in \propsetOf \phi} \Big(
\readable{p} \testpdl \ndet \big( \lnot \readable{p} \testpdl ; (\assgntopV{p} \ndet \assgnbotV{p}) \big) 
\Big)\right] \phi \testpdl $,
\item
all occurrences of $ \pi_1 \pll \pi_2 $ with $\progFlatten(\pi_1, \pi_2)$,
\item
all occurrences of $\assgntopV p$ with $\writable p ? ; \assgntopV p$,
\item
all occurrences of $\assgnbotV p$ with $\writable p ? ; \assgnbotV p$,
\item
all occurrences of $\assgntopR{p}$ with $\assgntop{ \readable{p}}$,
\item
all occurrences of $\assgnbotR{p}$ with $\assgnbot{ \writable{p}} ; \assgnbot{ \readable{p}}$,
\item
all occurrences of $\assgntopW{p}$ with $\assgntop{ \readable{p}} ; \assgntop{ \writable{p}}$,
\item
all occurrences of $\assgnbotW{p}$ with $\assgnbot{ \writable{p}}$.
\end{enumerate}
Let $t(\pi)$ be the translation of the \DlpaPll program $\pi$ and
$t(\phi)$ the translation of the \DlpaPll formula $\phi$.
Remember that when $t(\pi)$ and $t(\phi)$ are interpreted in \Dlpa, the variables $\readable p$ and $\writable{p}$ 
are considered to be arbitrary propositional variables. 

\begin{lemma}\label{lem:dlpatradcorrect}
For all \DlpaPll programs $\pi$ and formula $\phi$, and all \DlpaPll models $\modl$ and $\modl'$ such that $\modl \modinter P = \modl$,
\begin{align*}
  \modl \intPgm{\pi} \modl' &\text{ if and only if } \valuset^+ \intPgm{t(\pi)}^{\Dlpa} \valuset'^+ \text{, and} \\
  \modl \models \phi        &\text{ if and only if } \valuset^+ \models_{\Dlpa} t(\phi)
\end{align*}
where $\valuset^+  = \valuset  \cup \readOf{\readset } \cup \writeOf{\writeset }$
and   $\valuset'^+ = \valuset' \cup \readOf{\readset'} \cup \writeOf{\writeset'}$.
\end{lemma}
\begin{proof}[sketch]
The proof is by simultaneous induction on the size of $\pi$ or $\phi$.
The cases for endogenous test and parallel composition are similar to those in the proof of Theorem~\ref{theo:axiomatisationsound}.
The other cases are straightforward.
\end{proof}

The following theorem is a direct corollary of the previous lemma.

\begin{theorem}
A \DlpaPll formula $\phi$ is \DlpaPll-satisfiable if and only if the \Dlpa formula
$t(\phi) \land \bigwedge_{p \in \propsetOf \phi } (\writable{p} \limp \readable p)$ 
is \Dlpa satisfiable.
\end{theorem}

We can now state our complexity results.

\begin{lemma}\label{lem:polynomial}
The translation $t$ is polynomial.
\end{lemma}
\begin{proof}
It can easily be checked that for all $P$, $\progsplit(P)$, $\progStore(P)$, $\progOkChange(P)$ and $\progmerge(P)$ are linear in the size of $P$.
Therefore, $\progFlatten(\pi_1, \pi_2)$ is linear in the size of $\pi_1$ plus the size of $\pi_2$.
Similarly, 
$\left[\seqseq{p \in \propsetOf \phi} \Big(
\readable{p} \testpdl \ndet \big( \lnot \readable{p} \testpdl ; (\assgntopV{p} \ndet \assgnbotV{p}) \big) 
\Big)\right] \phi \testpdl $ is linear in the size of $\phi$.
All other translation expressions are clearly linear too.
Hence applying $t$ from the root of the syntax tree to its leaves, $t(\phi)$ can be computed in time polynomial in the size of $\phi$.
\end{proof}

\begin{theorem}
\DlpaPll model and satisfiability checking are both PSPACE complete.
\end{theorem}
\begin{proof}
First, PSPACE membership of \DlpaPll model and satisfiability checking follows from Lemmas~\ref{lem:dlpatradcorrect} and~\ref{lem:polynomial}. 
Second, since the language of \DlpaPll contains that of \Dlpa and since model and satisfiability checking are PSPACE hard for the latter~\cite{BalbianiHST14}, it follows that \DlpaPll model and satisfiability checking are PSPACE hard, too.
\end{proof}

\section{Discussion and Conclusion}\label{sec:conclusion}

We have added to Dynamic Logic of Propositional Assignments \Dlpa
a parallel composition operator in the spirit of separation logics. 
Our semantics augments \Dlpa valuations by readability and writability information. 
We have provided an axiomatisation in terms of a complete set of reduction axioms. 
Our reduction to \Dlpa ensures decidability. 
We have also proved PSPACE complexity via a polynomial translation to \Dlpa. 

We have adopted a stricter stance on race conditions than in~\cite{HerzigEtal-Ijcai19} 
where e.g.~the program $\assgntop p \pll \assgntop p$ is executable. 
Let us briefly compare these two semantics. 
In our case, the intuition is 
that $\pi_1 \pll \pi_2$ is executable if any interleaving of the components of $\pi_1$ and $\pi_2$ is executable, and 
that any of these interleavings leads to the same outcome. 
This is not guaranteed in the approach of \cite{HerzigEtal-Ijcai19}, which is motivated by parallel planning. 
There, it is generally considered that two actions that are executed in parallel should not interfere \cite{DBLP:journals/ai/BlumF97}:
they should not have conflicting effects and there should be no cross-interaction, where the second condition means that the effect of one action should not destroy the precondition of the other, and vice versa. 
For example, in a world of blocks the actions 
\begin{align*}
\mathsf{liftLeft}(b) &= \mathsf{OnTable}(b) ? ; \assgnbot{ \mathsf{OnTable}(b) } ; \assgntop{ \mathsf{HoldsLeft}(b) }  ,
\\
\mathsf{liftRight}(b) &= \mathsf{OnTable}(b) ? ; \assgnbot{ \mathsf{OnTable}(b) }  ; \assgntop{ \mathsf{HoldsRight}(b) } 
\end{align*}
of lifting block $b$ with the left robot arm and with the right robot arm have cross-interaction because 
one of the effects of $\mathsf{liftLeft}(b)$ is $\lnot \mathsf{OnTable}(b)$, which makes the precondition $\mathsf{pre}(\mathsf{liftRight}(b)) = \mathsf{OnTable}(b)$ of $\mathsf{lifRight}(b)$ false.
The more liberal semantics of \cite{HerzigEtal-Ijcai19} makes that it is not enough to describe the parallel execution of actions $\pi_1,\ldots,\pi_n$ as a step of a parallel plan by $\pi_1 \pll \ldots \pll \pi_n$. 
Instead, the absence of cross-interactions has to be checked `by hand', namely by explicitly inserting a \Dlpa test after each $\pi_i$ that the preconditions of the other actions are not violated: a step of a parallel plan is described by the program.
$$ \big( \pi_1 ; \bigwedge_{j \neq 1} \mathsf{pre}(\pi_j) ? \big) \pll \ldots \pll 
   \big( \pi_n ; \bigwedge_{j \neq n} \mathsf{pre}(\pi_j) ? \big) . $$
This is not necessary in our semantics where the absence of cross interaction between actions that are performed in parallel is `built-in'. 
Indeed, the parallel composition 
$\mathsf{liftLeft}(b) \pll \mathsf{liftRight}(b)$ 
is not executable in \DlpaPll: each subprogram requires writability of $\mathsf{OnTable}(b)$ to be executable. 
Our logic \DlpaPll can therefore be expected to provide a more appropriate base for parallel planning. 

The language of the extension of \Dlpa of \cite{HerzigEtal-Ijcai19} also contains an operator of inclusive nondeterministic composition, noted $\sqcup$. 
While the standard inclusive nondeterministic composition $\pi_1 \ndet \pi_2$ of \Pdl and \DlpaPll is read 
``do either $\pi_1$ or $\pi_2$'', the program $\pi_1 \sqcup \pi_2$ is read 
``do $\pi_1$ or $\pi_2$ \emph{or both}''. 
It has the same semantics as $\pi_1 \ndet \pi_2 \ndet (\pi_1 \pll \pi_2)$ and hence does not add expressivity. 
It is shown in \cite{HerzigEtal-Ijcai19} that it does not increase succinctness either.
This is proved by a polynomial reduction that uses the same `flattening' programs as the reduction of parallel composition. 
These programs are similar to our programs in Table~\ref{fig:useful_dlpa}, we therefore expect that inclusive nondeterministic composition does not increase the succinctness of the language of \DlpaPll either. 

The mathematical properties of \DlpaPll compare favourably with 
the high complexity or even undecidability of the 
other extensions of dynamic logic by a separating parallel composition operator that were proposed in the literature
\cite{BalbianiT14,Boudou16}.
Just as ours, the latter line of work is in the spirit of separation logic, having splitting and merging operations that are defined on system states. 
The axiomatisation that was introduced and studied in \cite{DBLP:journals/logcom/BalbianiB18} is restricted to the star-free fragment and the authors had to add propositional quantifiers in order to make parallel composition definable. 
This contrasts with the simplicity of our axiomatisation of \DlpaPll that we obtained by adding reduction axioms to the axiomatisation of \Dlpa. 
This can be related to the fact that propositional quantifiers can be expressed in \Dlpa: 
$\exists p \phi$ is equivalent to $\ldia{ \assgntopV p \ndet \assgnbotV p } \phi$ and 
$\forall p \phi$ is equivalent to $\lbox{ \assgntopV p \ndet \assgnbotV p } \phi$.
Just as \Dlpa can be viewed as an instance of PDL%
---the interpretation of atomic programs moves from PDL's abstract relation between states to concrete updates of valuations\mbox{---,}
\DlpaPll can be viewed as an instance of the logic of \cite{DBLP:journals/entcs/BenevidesFV11}
where the interpretation of parallel composition no longer resorts to
an abstract relation $\star$ associating three states, 
but instead has concrete functions that split and merge valuations and that are constrained by readability and writability information.

\section{Acknowledgements}
The paper benefited from comments and remarks from the reviewers as well as from the attendees of DaL\'i 2019, 
in particular Alexandru Baltag, Raul Fervari, Rainer H\"ahnle and Dexter Kozen. 
We would like to particularly thank the three reviewers of JLAMP who provided detailed and well-informed reviews. 
We did our best take all these comments into account. 

Andreas Herzig was partially supported by TAILOR, a project funded by EU 
Horizon 2020 research and innovation programme under GA No 952215.
Nicolas Troquard was supported by 
UNIBZ CRC 2019 project IN2092, Computations in Resource Aware Systems (CompRAS).

\appendix
\section{Appendix: Associativity of Parallel Composition}

\newcommand{\modla}{\mathsf a}
\newcommand{\readseta}{\mathsf {Rda}}
\newcommand{\writeseta}{\mathsf {Wra}}
\newcommand{\valuseta}{\mathsf {Va}}

Alongside commutativity, associativity is a desirable property of a parallel composition operator. While the operator defined in Section~\ref{sec:interpretation} is clearly commutative, whether it is associative does not strike the eye. In fact, a detailed proof of it is rather cumbersome. 
We present it here.

\begin{proposition}\label{prop:associativity}
$\modl \intPgm{ \pi_1 \pll (\pi_2 \pll \pi_3) } \modl'$ iff $\modl \intPgm{ (\pi_1 \pll \pi_2) \pll \pi_3 } \modl'$.
\end{proposition}
\begin{proof}

  \begin{figure}
    \centering
    \begin{tabular}{lcl}
      $\exists \modl_1, \modl_*, \modl_2, \modl_3, \modl'_1, \modl'_2, \modl'_3, \modl'_*$': &  ~~~~~\emph{iff}~~~~~&    $\exists \modla_\#, \modla_3, \modla_1, \modla_2, \modla_1', \modla_2', \modla_\#', \modla_3'$:\\
      \begin{tikzpicture}[>=latex', join=bevel, initial text = , every node/.style=, scale=1.2]
  \node (m) at (0bp, 0bp) {$\modl$};

  \node (m1) at (10bp, 30bp) {$\modl_1$};
  \node (mast) at (10bp, -30bp) {$\modl_*$};

  \node (m1prime) at (100bp, 30bp) {$\modl'_1$};
  \node (m2) at (20bp, -5bp) {$\modl_2$};  
  \node (m3) at (20bp, -55bp) {$\modl_3$};  

  \node (mastprime) at (100bp, -30bp) {$\modl'_*$};
  \node (m2prime) at (90bp, -5bp) {$\modl'_2$};  
  \node (m3prime) at (90bp, -55bp) {$\modl'_3$};  

  \node (mprime) at (110bp, 0bp) {$\modl'$};
  \draw[thick, dashed, ->] (m) to node [] {} (m1);
  \draw[thick, dashed, ->] (m) to node [] {} (mast);

  \draw[thick, dashed, ->] (mast) to node [] {} (m2);
  \draw[thick, dashed, ->] (mast) to node [] {} (m3);

  \draw[thick, ->] (m1) to node [below] {$\pi_1$} (m1prime);  
  \draw[thick, ->] (m2) to node [below] {$\pi_2$} (m2prime);  
  \draw[thick, ->] (m3) to node [above] {$\pi_3$} (m3prime);  

  \draw[thick, dotted,->] (m2prime) to node [] {} (mastprime);
  \draw[thick, dotted,->] (m3prime) to node [] {} (mastprime);

  \draw[thick, dotted, ->] (m1prime) to node [] {} (mprime);  
  \draw[thick, dotted,->] (mastprime) to node [] {} (mprime);  
  
\end{tikzpicture}
      & ~ &
\begin{tikzpicture}[>=latex', join=bevel, initial text = , every node/.style=, scale=1.2]
  \node (m) at (0bp, 0bp) {$\modl$};

  \node (a3) at (10bp, -30bp) {$\modla_3$};
  \node (msharp) at (10bp, 30bp) {$\modla_\#$};

  \node (a3prime) at (100bp, -30bp) {$\modla'_3$};
  \node (a2) at (20bp, 5bp) {$\modla_2$};  
  \node (a1) at (20bp, 55bp) {$\modla_1$};  

  \node (msharpprime) at (100bp, 30bp) {$\modla'_\#$};
  \node (a2prime) at (90bp, 5bp) {$\modla'_2$};  
  \node (a1prime) at (90bp, 55bp) {$\modla'_1$};  

  \node (mprime) at (110bp, 0bp) {$\modl'$};
  \draw[thick, dashed, ->] (m) to node [] {} (a3);
  \draw[thick, dashed, ->] (m) to node [] {} (msharp);

  \draw[thick, dashed, ->] (msharp) to node [] {} (a2);
  \draw[thick, dashed, ->] (msharp) to node [] {} (a1);

  \draw[thick, ->] (a3) to node [above] {$\pi_3$} (a3prime);  
  \draw[thick, ->] (a2) to node [above] {$\pi_2$} (a2prime);  
  \draw[thick, ->] (a1) to node [below] {$\pi_1$} (a1prime);  

  \draw[thick, dotted, ->] (a2prime) to node [] {} (msharpprime);
  \draw[thick, dotted, ->] (a1prime) to node [] {} (msharpprime);

  \draw[thick, dotted, ->] (a3prime) to node [] {} (mprime);  
  \draw[thick, dotted, ->] (msharpprime) to node [] {} (mprime);  
  
\end{tikzpicture}
    \end{tabular}
    \caption{\label{fig:illustration-associativity} Illustration of associativity. Visual aid for the proof of Proposition~\ref{prop:associativity}.}
  \end{figure}
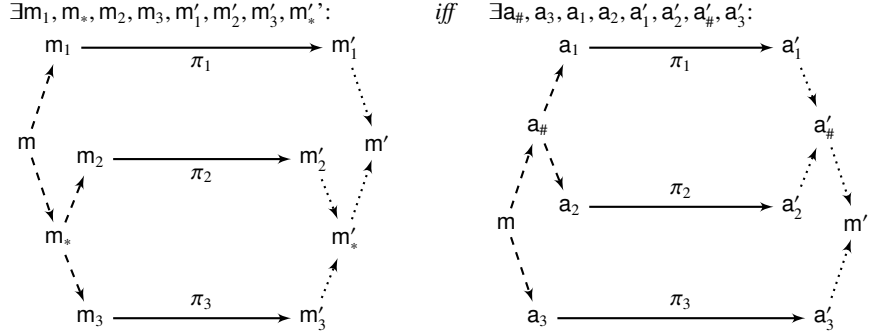
  





Figure~\ref{fig:illustration-associativity} illustrates precisely what we are going to show.
Suppose $\modl = \tuple{\readset, \writeset, \valuset}$ and $\modl' = \tuple{\readset', \writeset', \valuset'}$.

\paragraph{Left-hand side.} $\modl \intPgm{ \pi_1 \pll (\pi_2 \pll \pi_3) } \modl'$.

There are $\modl_1, \modl_*, \modl'_1, \modl'_*$ such that (with, $\modl_1 = \tuple{\readset_1, \writeset_1, \valuset_1}$, $\modl'_1 = \tuple{\readset'_1, \writeset'_1, \valuset'_1}$, $\modl_* = \tuple{\readset_*, \writeset_*, \valuset_*}$, $\modl'_* = \tuple{\readset'_*, \writeset'_*, \valuset'_*}$):
\begin{enumerate}
\item\label{i:split-merge1}\label{i:lhsfirst} $\splt{\modl}{\modl_1} {\modl_*}$
  and $\mrg{\modl'}{\modl'_1} {\modl'_*}$ 
\item $\modl_1 \intPgm{ \pi_1 } \modl'_1$
\item $\readset_1 = \readset'_1$
  and $\writeset_1 = \writeset'_1$
  and $\valuset_1 \setminus \writeset_1 = \valuset'_1 \setminus \writeset'_1$
\item\label{i:pi2pi3} $\modl_* \intPgm{ \pi_2 \pll \pi_3 } \modl'_*$ 
\item $\readset_* = \readset'_*$
  and $\writeset_* = \writeset'_*$ 
  and $\valuset_* \setminus \writeset_* = \valuset'_* \setminus \writeset'_*$
\end{enumerate}
Item~\ref{i:pi2pi3} is equivalent to: there are $\modl_2, \modl_3, \modl'_2, \modl'_3$ such that  (with $\modl_2 = \tuple{\readset_2, \writeset_2, \valuset_2}$, $\modl'_2 = \tuple{\readset'_2, \writeset'_2, \valuset'_2}$, $\modl_3 = \tuple{\readset_3, \writeset_3, \valuset_3}$, $\modl'_3 = \tuple{\readset'_3, \writeset'_3, \valuset'_3}$):
\begin{enumerate}[resume]
\item\label{i:split-merge2} $\splt{\modl_*}{\modl_2} {\modl_3} $ and $\mrg{\modl_*'}{\modl'_2} {\modl'_3} $,
\item $\modl_2 \intPgm{ \pi_2 } \modl'_2$, 
\item $\readset_2 = \readset'_2 $ and $\writeset_2 = \writeset'_2 $ and $\valuset_2 \setminus \writeset_2 = \valuset'_2 \setminus \writeset'_2$,
\item $\modl_3 \intPgm{ \pi_3 } \modl'_3$, 
\item $\readset_3 = \readset'_3 $ and $\writeset_3 = \writeset'_3 $ and $\valuset_3 \setminus \writeset_3 = \valuset'_3 \setminus \writeset'_3$.
\end{enumerate}
Item~\ref{i:split-merge1} is $\splt{\modl}{\modl_1} {\modl_*} $ and $\mrg{\modl'}{\modl'_1} {\modl'_*} $ iff:
\begin{enumerate}[resume]
\item
  $\writeset_1 \cap \readset_* = \writeset_* \cap \readset_1 = \emptyset$ ($\modl_1$ and $\modl_*$ are RW-compatible),
  \item 
    $\readset = \readset_1 \cup \readset_* $, and $\writeset = \writeset_1 \cup \writeset_*$, and $\valuset = \valuset_1 = \valuset_*$,

    
    \item $\writeset'_1 \cap \readset'_* = \writeset'_* \cap \readset'_1 = \emptyset$ ($\modl'_1$ and $\modl'_*$ are RW-compatible),
    \item $\readset' = \readset'_1 \cup \readset'_*$, $\writeset' = \writeset'_1 \cup \writeset'_*$, and $\valuset'_1 \setminus \writeset' = \valuset'_* \setminus \writeset'$, and $\valuset' = (\valuset'_1 \cap \writeset'_1) \cup (\valuset'_* \cap \writeset'_*) \cup (\valuset'_1 \cap \valuset'_*) $. 
\end{enumerate}
Item~\ref{i:split-merge2} is $\splt{\modl}{\modl_2} {\modl_3} $ and $\mrg{\modl'}{\modl'_2} {\modl'_3} $ iff:
\begin{enumerate}[resume]
\item
  $\writeset_2 \cap \readset_3 = \writeset_3 \cap \readset_2 = \emptyset$ ($\modl_2$ and $\modl_3$ are RW-compatible),
  \item 
    $\readset_* = \readset_2 \cup \readset_3 $, and $\writeset_* = \writeset_2 \cup \writeset_3$, and $\valuset_* = \valuset_2 = \valuset_3$,

    
    \item $\writeset'_2 \cap \readset'_3 = \writeset'_3 \cap \readset'_2 = \emptyset$ ($\modl'_2$ and $\modl'_3$ are RW-compatible),
    \item\label{i:lhslast} $\readset_*' = \readset'_2 \cup \readset'_3$, $\writeset_*' = \writeset'_2 \cup \writeset'_3$, and $\valuset'_2 \setminus \writeset_*' = \valuset'_3 \setminus \writeset_*'$, and $\valuset_*' = (\valuset'_2 \cap \writeset'_2) \cup (\valuset'_3 \cap \writeset'_3) \cup (\valuset'_2 \cap \valuset'_3) $. 
\end{enumerate}

\paragraph{Right-hand side.} $\modl \intPgm{ (\pi_1 \pll \pi_2) \pll \pi_3 } \modl'$.

There are 
$\modla_\# = \tuple{\readseta_\#, \writeseta_\#, \valuseta_\#}$, $\modla'_\# = \tuple{\readseta'_\#, \writeseta'_\#, \valuseta'_\#}$, $\modla_3 = \tuple{\readseta_3, \writeseta_3, \valuseta_3}$ and $\modla'_3 = \tuple{\readseta'_3, \writeseta'_3, \valuseta'_3}$
such that:
\begin{enumerate}[resume]
\item\label{i:split-merge3}\label{i:rhsfirst} $\splt{\modl}{\modla_\#} {\modla_3}$
  and $\mrg{\modl'}{\modla'_\#} {\modla'_3}$, 
\item\label{i:pi1pi2} $\modla_\# \intPgm{ \pi_1 \pll \pi_2} \modla'_\#$,
\item $\readseta_\# = \readseta'_\#$
  and $\writeseta_\# = \writeseta'_\#$
  and $\valuseta_\# \setminus \writeseta_\# = \valuseta'_\# \setminus \writeseta'_\#$,
\item $\modla_3 \intPgm{ \pi_3 } \modla'_3$, 
\item $\readseta_3 = \readseta'_3$
  and $\writeseta_3 = \writeseta'_3$ 
  and $\valuseta_3 \setminus \writeseta_3 = \valuseta'_3 \setminus \writeseta'_3$.
\end{enumerate}
Item~\ref{i:pi1pi2} is equivalent to:
there are $\modla_1, \modla_2, \modla'_1, \modla'_2$ such that 
(with, $\modla_1 = \tuple{\readseta_1, \writeseta_1, \valuseta_1}$, $\modla'_1 = \tuple{\readseta'_1, \writeseta'_1, \valuseta'_1}$, $\modla_2 = \tuple{\readseta_2, \writeseta_2, \valuseta_2}$, $\modla'_2 = \tuple{\readseta'_2, \writeseta'_2, \valuseta'_2}$):
\begin{enumerate}[resume]
\item\label{i:split-merge4} $\splt{\modla_\#}{\modla_1} {\modla_2}$
  and $\mrg{\modla'_\#}{\modla'_1} {\modla'_2}$, 
\item $\modla_1 \intPgm{ \pi_1} \modla'_1$,
\item $\readseta_1 = \readseta'_1$,
  and $\writeseta_1 = \writeseta'_1$
  and $\valuseta_1 \setminus \writeseta_1 = \valuseta'_1 \setminus \writeseta'_1$,
\item $\modla_2 \intPgm{ \pi_2 } \modla'_2$, 
\item $\readseta_2 = \readseta'_2$
  and $\writeseta_2 = \writeseta'_2$ 
  and $\valuseta_2 \setminus \writeseta_2 = \valuseta'_2 \setminus \writeseta'_2$.
\end{enumerate}
Item~\ref{i:split-merge3} is $\splt{\modl}{\modla_\#} {\modla_3}$ and $\mrg{\modl'}{\modla'_\#} {\modla'_3}$ iff:
\begin{enumerate}[resume]
\item
  $\writeseta_\# \cap \readseta_3 = \writeseta_3 \cap \readseta_\# = \emptyset$ ($\modla_\#$ and $\modla_3$ are RW-compatible),
  \item 
    $\readset = \readseta_\# \cup \readseta_3 $, and $\writeset = \writeseta_\# \cup \writeseta_3$, and $\valuset = \valuseta_\# = \valuseta_3$,

    
    \item $\writeseta'_\# \cap \readseta'_3 = \writeseta'_3 \cap \readseta'_\# = \emptyset$ ($\modla'_\#$ and $\modla'_3$ are RW-compatible),
    \item $\readset' = \readseta'_\# \cup \readseta'_3$, $\writeset' = \writeseta'_\# \cup \writeseta'_3$, and $\valuseta'_\# \setminus \writeset' = \valuseta'_3 \setminus \writeset'$, and $\valuset' = (\valuseta'_\# \cap \writeseta'_\#) \cup (\valuseta'_3 \cap \writeseta'_3) \cup (\valuseta'_\# \cap \valuseta'_3) $. 
\end{enumerate}
Item~\ref{i:split-merge4} is $\splt{\modla_\#}{\modla_1} {\modla_2}$ and $\mrg{\modla'_\#}{\modla'_1} {\modla'_2}$ iff:
\begin{enumerate}[resume]
\item
  $\writeseta_1 \cap \readseta_2 = \writeseta_2 \cap \readseta_1 = \emptyset$ ($\modla_1$ and $\modla_2$ are RW-compatible),
  \item 
    $\readseta_\# = \readseta_1 \cup \readseta_2 $, and $\writeseta_\# = \writeseta_1 \cup \writeseta_2$, and $\valuseta_\# = \valuseta_1 = \valuseta_2$,

    
    \item $\writeseta'_1 \cap \readseta'_2 = \writeseta'_2 \cap \readseta'_1 = \emptyset$ ($\modla'_1$ and $\modla'_2$ are RW-compatible),
    \item\label{i:rhslast} $\readseta_\#' = \readseta'_1 \cup \readseta'_2$, $\writeseta_\#' = \writeseta'_1 \cup \writeseta'_2$, and $\valuseta'_1 \setminus \writeseta_\#' = \valuseta'_2 \setminus \writeseta_\#'$, and $\valuseta_\#' = (\valuseta'_1 \cap \writeseta'_1) \cup (\valuseta'_2 \cap \writeseta'_2) \cup (\valuseta'_1 \cap \valuseta'_2) $. 
\end{enumerate}

\paragraph{Left to right.}

Suppose lhs. We define:
\begin{itemize} 
\item $\modla_\# = \tuple{\readseta_\#, \writeseta_\#, \valuseta_\#} = \tuple{\readset_1 \cup \readset_2, \writeset_1 \cup \writeset_2, \valuset}$,
\item $\modla_1 = \tuple{\readseta_1, \writeseta_1, \valuseta_1} = \modl_1 = \tuple{\readset_1, \writeset_1, \valuset}$,
\item $\modla_2 = \tuple{\readseta_2, \writeseta_2, \valuseta_2} = \modl_2 = \tuple{\readset_2, \writeset_2, \valuset}$,
\item $\modla_3 = \tuple{\readseta_3, \writeseta_3, \valuseta_3} = \modl_3 = \tuple{\readset_3, \writeset_3, \valuset}$,
\item $\modla'_\# = \tuple{\readseta'_\#, \writeseta'_\#, \valuseta'_\#} = \tuple{\readset'_1 \cup \readset'_2, \writeset'_1 \cup \writeset'_2, (\valuset'_1 \cap \writeset_1) \cup (\valuset'_2 \cap \writeset_2) \cup (\valuset'_1 \cap \valuset'_2)}$,
\item $\modla'_1 = \tuple{\readseta'_1, \writeseta'_1, \valuseta'_1) = \modl'_1 = (\readset'_1, \writeset'_1, \valuset'_1}$,
\item $\modla'_2 = \tuple{\readseta'_2, \writeseta'_2, \valuseta'_2} = \modl'_2 = \tuple{\readset'_2, \writeset'_2, \valuset'_2}$,
\item $\modla'_3 = \tuple{\readseta'_3, \writeseta'_3, \valuseta'_3} = \modl'_3 = \tuple{\readset'_3, \writeset'_3, \valuset'_3}$. 
\end{itemize}

We must show that these models satisfy all the properties from~\ref{i:rhsfirst} through~\ref{i:rhslast}.
\begin{itemize}
\item[19] if and only if
  \begin{itemize}
  \item[29]  $\writeseta_\# \cap \readseta_3 = \emptyset$,
    $\writeseta_3 \cap \readseta_\# = \emptyset$. It holds because:
    \begin{itemize}
    \item $\writeseta_\# \cap \readseta_3 = (\writeset_1 \cup \writeset_2) \cap \readset_3 = (\writeset_1 \cap \readset_3) \cup (\writeset_2 \cap \readset_3)$.
      By 15, we have $\writeset_2 \cap \readset_3 = \emptyset$.
      By 16, we have $\readset_3 \subseteq \readset_*$.
      By 11, we have $\writeset_1 \cap \readset_* = \emptyset$. So, $\writeset_1 \cap \readset_3 = \emptyset$.
    \item $\writeseta_3 \cap \readseta_\# = \writeset_3 \cap (\readset_1 \cup \readset_2) = (\writeset_3 \cap \readset_1) \cup (\writeset_3 \cap \readset_2)$.
      By 15, we have $\writeset_3 \cap \readset_2 = \emptyset$.
      By 11, we have $\writeset_* \cap \readset_1 = \emptyset$.
      By 16, we have $\writeset_3 \subseteq \writeset_*$.
      So, $\writeset_3 \cap \readset_1 = \emptyset$.
    \end{itemize}
  
  \item[30]
$\readseta_\# \cup \readseta_3 = (\readset_1 \cup \readset_2) \cup \readset_3 = \readset_1 \cup \readset_* = \readset$ (definition and 16 and 12).
    $\writeseta_\# \cup \writeseta_3 = (\writeset_1 \cup \writeset_2) \cup \writeset_3 = \writeset_1 \cup \writeset_* = \writeset$ (definition and 16 and 12).
    $\valuset = \valuseta_\# = \valuseta_3$ (definition).
    
  \item[31] $\writeseta'_\# \cap \readseta'_3 = \writeseta'_3 \cap \readseta'_\# = \emptyset$. It holds because:
    \begin{itemize}
    \item $\writeseta'_\# \cap \readseta'_3 = (\writeset'_1 \cup \writeset'_2) \cap \readset'_3$ by definition. It is equal to $(\writeset_1 \cap \readset_3) \cup (\writeset_2 \cap \readset_3)$, by 3, 8, 10. We have $\writeset_1 \cap \readset_* = \emptyset$ (11), and $\readset_3 \subseteq \readset_*$ (16). So $\writeset_1 \cap \readset_3 = \emptyset$.
      We have $\writeset_2 \cap \readset_3 = \emptyset$ (15).
    \item $\writeseta'_3 \cap \readseta'_\# = \writeseta'_3 \cap (\readset'_1 \cup \readset'_2)$ by definition. It is equal to $(\writeseta_3 \cap \readset_1) \cup (\writeseta_3 \cap \readset_2)$, by 10, 3, 8. We have $\writeset_3 \subseteq \writeset_*$ (16) and $\writeset_* \cap \readset_1 = \emptyset$ (11).
      So $\writeseta_3 \cap \readset_1 = \emptyset$.
      We have $\writeseta_3 \cap \readset_2 = \emptyset$ (15).
    \end{itemize}

  \item[32]
    \begin{itemize}
    \item Instrumental claims:
      \begin{itemize}
  \item(claim~1) $\writeset'_* = \writeset_*$, by 16, 7, 9 and 18.
  \item(claim~2) $\writeset = \writeset'$, by 12, 3, claim~1, 14.
  \item(claim~3) $\writeset_1 \subseteq \writeset'$, by claim~2, and 12.
  \item(claim~4.1) $\writeset_2 \subseteq \writeset'$, by claim~2, 16, and 12.
  \item(claim~4.2) $\writeset_3 \subseteq \writeset'$, by claim~2, 16, and 12.    
  \item(claim~5) $\valuset'_1 \setminus \writeset' = \valuset_1 \setminus \writeset'$, by 3, 12, and claim~2.
  \item(claim~6) $\valuset'_2 \setminus \writeset' = \valuset_2 \setminus \writeset'$, by 8, 16, and 12.
  \item(claim~7) $\valuset'_3 \setminus \writeset' = \valuset_3 \setminus \writeset'$, by 10, 16, and 12.
  \item(claim~8) $\writeset_1 \cap \writeset_2 = \emptyset$, by 11, 12, the `write-set included in read-set' model constraint, and 16.
      \end{itemize}
    \item $\readset' = \readseta'_\# \cup \readseta'_3$ and $\writeset' = \writeseta'_\# \cup \writeseta'_3$ hold by definition, 16 and 12. 

    \item $\valuseta'_\# \setminus \writeset' = \valuseta'_3 \setminus \writeset'$ holds because:
      \begin{itemize}
        \item $\valuseta'_\# \setminus \writeset' = (\valuset_1' \cup \valuset_2') \setminus \writeset'$ by definitions, claim~3, and claim~4.1. By claim~5 and claim~6, it is equal to $(\valuset_1 \cup \valuset_2) \setminus \writeset'$, which by 12 and 16 is $\valuset \setminus \writeset'$.
          \item
      Moreover, $\valuseta'_3 \setminus \writeset' = \valuset'_3 \setminus \writeset'$ by definition. It is equal to $\valuset_3 \setminus \writeset'$ by claim~7, and to $\valuset \setminus \writeset'$ by 12 and 16.
      \end{itemize}
    \item $\valuset' = (\valuseta'_\# \cap \writeseta'_\#) \cup (\valuseta'_3 \cap \writeseta'_3) \cup (\valuseta'_\# \cap \valuseta'_3)$  holds because:
      \begin{itemize}
      \item $\valuseta'_\# \cap \writeseta'_\# = ((\valuset'_1 \cap \writeset'_1) \cup (\valuset'_2 \cap \writeset'_2) \cup (\valuset'_1 \cap \valuset'_2)) \cap (\writeset'_1 \cup \writeset'_2)$ by definition, 3, and 8. We get $(\valuset'_1 \cap \writeset'_1 \cap \writeset'_1) \cup (\valuset'_1 \cap \writeset'_1 \cap \writeset'_2) \cup (\valuset'_2 \cap \writeset'_2 \cap \writeset'_1) \cup (\valuset'_2 \cap \writeset'_2 \cap \writeset'_2) \cup (\valuset'_1 \cap \valuset'_2 \cap \writeset'_1) \cup (\valuset'_1 \cap \valuset'_2 \cap \writeset'_2)$. With elementary set theory simplifications and claim~8, we obtain $(\valuset'_1 \cap \writeset'_1) \cup (\valuset'_2 \cap \writeset'_2)$.
\item $\valuseta'_3 \cap \writeseta'_3 = \valuset'_3 \cap \writeset'_3$.
\item $\valuseta'_\# \cap \valuseta'_3 ((\valuset'_1 \cap \writeset'_1) \cup (\valuset'_2 \cap \writeset'_2) \cup (\valuset'_1 \cap \valuset'_2)) \cap \valuset'_3$ by definition, 3, and 8. We get $(\valuset'_1 \cap \writeset'_1 \cap \valuset'_3) \cup (\valuset'_2 \cap \writeset'_2 \cap \valuset'_3) \cup (\valuset'_1 \cap \valuset'_2 \cap \valuset'_3)$.
\item The rhs quantity $(\valuseta'_\# \cap \writeseta'_\#) \cup (\valuseta'_3 \cap \writeseta'_3) \cup (\valuseta'_\# \cap \valuseta'_3)$ is then equal to
  $(\valuset'_1 \cap \writeset'_1) \cup (\valuset'_2 \cap \writeset'_2) \cup
  (\valuset'_3 \cap \writeset'_3) \cup (\valuset'_1 \cap \valuset'_2 \cap \valuset'_3)$.
\item Moreover, by 14, we have $\valuset' = (\valuset'_1 \cap \writeset'_1) \cup (\valuset'_* \cap \writeset'_*) \cup (\valuset'_1 \cap \valuset'_*)$.
\item By 18, $\valuset'_* = (\valuset'_2 \cap \writeset'_2) \cup (\valuset'_3 \cap \writeset'_3) \cup (\valuset'_1 \cap \valuset'_2)$.
\item By claim~1, $\writeset'_* = \writeset_*$, which by 16 is $\writeset_2 \cup \writeset_3$ which by 8 and 10 is $\writeset'_2 \cup \writeset'_3$. So $\writeset'_2 \subseteq \writeset'_*$ and $\writeset'_3 \subseteq \writeset'_*$.
\item So $(\valuset'_* \cap \writeset'_*) = (\valuset'_2 \cap \writeset'_2) \cup (\valuset'_3 \cap \writeset'_3) \cup (\valuset'_1 \cap \valuset'_2 \cap \writeset'_*)$, with $\valuset'_1 \cap \valuset'_2 \cap \writeset'_* = \valuset'_1 \cap \valuset'_2 \cap (\writeset'_2 \cup \writeset'_3)$ by 18, which is $(\valuset'_2 \cap \valuset'_3 \cap \writeset'_2) \cup (\valuset'_2 \cap \valuset'_3 \cap \writeset'_3)$. So $(\valuset'_* \cap \writeset'_*) = (\valuset'_2 \cap \writeset'_2) \cup (\valuset'_3 \cap \writeset'_3)$.
\item Also $(\valuset'_1 \cap \valuset'_*) = (\valuset'_1 \cap \valuset'_2 \cap \writeset'_2) \cup (\valuset'_1 \cap \valuset'_3 \cap \writeset'_3) \cup (\valuset'_1 \cap \valuset'_2 \cap \valuset'_3)$, with the first two disjuncts included in $\valuset'_* \cap \writeset'_*$.
\item The lhs quantity $V'$ is then equal to $(\valuset'_1 \cap \writeset'_1) \cup (\valuset'_2 \cap \writeset'_2) \cup
  (\valuset'_3 \cap \writeset'_3) \cup (\valuset'_1 \cap \valuset'_2 \cap \valuset'_3)$.
      \end{itemize}
    \end{itemize}

  \end{itemize}

\item[20] if and only if

  \begin{itemize}
  \item[24] if and only if
  \begin{itemize}
  \item[33] $\writeseta_1 \cap \readseta_2 = \emptyset$,
    $\writeseta_2 \cap \readseta_1 = \emptyset$. It holds, because:
    \begin{itemize}

    \item $\writeseta_1 \cap \readseta_2 = \writeset_1 \cap \readset_2$ by definition.
    From 16, $\readset_2 \subseteq \readset_*$. From 11, $\writeset_1 \cap \readset_* = \emptyset$. So $\writeset_1 \cap \readset_2 = \emptyset$.

    \item $\writeseta_2 \cap \readseta_1 = \writeset_2 \cap \readset_1$ by definition.
    From 16, $\writeset_2 \subseteq \writeset_*$. From 11, $\writeset_* \cap \readset_1 = \emptyset$. So $\writeset_2 \cap \readset_1 = \emptyset$.
    \end{itemize}
  \item[34] By definition.
  \item[35] By definition, 33, 3, and 8.
  \item[36] $\readseta_\#' = \readseta'_1 \cup \readseta'_2$, $\writeseta_\#' = \writeseta'_1 \cup \writeseta'_2$, and $\valuseta_\#' = (\valuseta'_1 \cap \writeseta'_1) \cup (\valuseta'_2 \cap \writeseta'_2) \cup (\valuseta'_1 \cap \valuseta'_2)$ by definition.
    Also, $\valuseta'_1 \setminus \writeseta_\#' = \valuseta'_2 \setminus \writeseta_\#'$ holds because:
    \begin{itemize}
    \item $\valuseta'_1 \setminus \writeseta_\#' = \valuset'_1 \setminus (\writeset'_1 \cup \writeset'_2)$, which by 3 and 8 is equal to $\valuset_1 \setminus (\writeset_1 \cup \writeset_2)$, which by 12 is equal to $\valuset \setminus (\writeset_1 \cup \writeset_2)$.
      \item Similarly, $\valuseta'_2 \setminus \writeseta_\#'$ is equal to $\valuset \setminus (\writeset_1 \cup \writeset_2)$ by 8, 3, 16 and 12.
    \end{itemize}
  \end{itemize}
\item[25] By definition, and 2.
\item[26] By definition, 3, and 12.
\item[27] By definition, and 7.
\item[28] By definition, 8, 16, and 12.
  \end{itemize}
\item[21] $\readseta_\# = \readseta'_\#$
  and $\writeseta_\# = \writeseta'_\#$ hold by definition, 3, and 8.
  Also by definition, 3, and 8, $\valuseta_\# \setminus \writeseta_\# = \valuseta'_\# \setminus \writeseta'_\#$ is equivalent to $\valuset \setminus (\writeset_1 \cup \writeset_2) = ((\valuset'_1 \cap \writeset'_1) \cup (\valuset'_2 \cap \writeset'_2) \cup (\valuset'_1 \cap \valuset'_2)) \setminus (\writeset_1 \cup \writeset_2)$. The right-hand-side simplifies into $(\valuset'_1 \cap \valuset'_2) \setminus (\writeset_1 \cup \writeset_2)$. Moreover, we have $\valuset'_1 \setminus \writeset'_1 = \valuset \setminus \writeset'_1$ (by 3 and 12) and $\valuset'_2 \setminus \writeset'_2 = \valuset \setminus \writeset'_2$ (by 8, 16, and 12). So we have $\valuset'_1 \setminus (\writeset'_1 \cup \writeset'_2) = \valuset \setminus (\writeset'_1 \cup \writeset'_2)$ and $\valuset'_2 \setminus (\writeset'_1 \cup \writeset'_2) = \valuset \setminus (\writeset'_1 \cup \writeset'_2)$. Hence, $(\valuset'_1 \cap \valuset'_2) \setminus (\writeset_1 \cup \writeset_2)$ simplifies into the left-hand-side $\valuset \setminus (\writeset_1 \cup \writeset_2)$.
\item[22] By definition, and 9.
\item[23] By definition, 10, 16, and 12.
\end{itemize}

\paragraph{Right to left.}

Suppose rhs. We define:
\begin{itemize}
\item $\modl_* = \tuple{\readset_*, \writeset_*, \valuset_*} = \tuple{\readseta_2 \cup \readseta_3, \writeseta_2 \cup \writeseta_3, \valuset}$,

\item $\modl_1 = \tuple{\readset_1, \writeset_1, \valuset_1} = \modla_1 = \tuple{\readseta_1, \writeseta_1, \valuseta}$,
\item $\modl_2 = \tuple{\readset_2, \writeset_2, \valuset_2} = \modla_2 = \tuple{\readseta_2, \writeseta_2, \valuseta}$,
\item $\modl_3 = \tuple{\readset_3, \writeset_3, \valuset_3} = \modla_3 = \tuple{\readseta_3, \writeseta_3, \valuseta}$,  

\item $\modl'_* = \tuple{\readset'_*, \writeset'_*, \valuset'_*} = \tuple{\readseta'_2 \cup \readseta'_3, \writeseta'_2 \cup \writeseta'_3, (\valuseta'_2 \cap \writeseta_2) \cup (\valuseta'_3 \cap \writeseta_3) \cup (\valuseta'_2 \cap \valuseta'_3)}$,
  
\item $\modl'_1 = \tuple{\readset'_1, \writeset'_1, \valuset'_1} = \modla'_1 = \tuple{\readseta'_1, \writeseta'_1, \valuseta'_1}$,
\item $\modl'_2 = \tuple{\readset'_2, \writeset'_2, \valuset'_2} = \modla'_2 = \tuple{\readseta'_2, \writeseta'_2, \valuseta'_2}$,
\item $\modl'_3 = \tuple{\readset'_3, \writeset'_3, \valuset'_3} = \modla'_3 = \tuple{\readseta'_3, \writeseta'_3, \valuseta'_3}$. 
\end{itemize}

We must show that these models satisfy all the properties from~\ref{i:lhsfirst} to~\ref{i:lhslast}. This is done routinely, analogously to the proof of the left-to-right direction above.
\end{proof}


\bibliographystyle{alpha}
\bibliography{biblio}

\end{document}